\documentclass[runningheads]{llncs}
\usepackage{amsmath,amssymb,amsfonts}
\usepackage{diagbox}
\usepackage{makecell} 
\usepackage{pifont}
\usepackage{algorithm,algorithmic}

\usepackage[usenames,dvipsnames]{xcolor}
\newcommand{\ignore}[1]{}
\newcommand{\ketbra}[2]{\left\lvert #1 \right\rangle \left\langle #2 \right\lvert}
\newcommand{\etal}{\textit{et al}.}
\usepackage{amsmath,color}
\usepackage{braket}
\usepackage{graphicx}
\usepackage{textcomp}
\newcommand{\NL}[1]{\textcolor{blue}{#1}}
\newcommand{\A}{T}
\newcommand{\w}{w}
\newcommand{\pe}{\mathfrak{p}}
\newcommand{\mc}{\mathcal}
\newcommand{\msf}{\mathsf}
\def\BibTeX{{\rm B\kern-.05em{\sc i\kern-.025em b}\kern-.08em
    T\kern-.1667em\lower.7ex\hbox{E}\kern-.125emX}}

    \newtheorem{innercustomlm}{\textbf{Lemma}}
\newenvironment{customlemma}[1]
  {\renewcommand\theinnercustomlm{#1}\innercustomlm}
  {\endinnercustomlm}
   \newtheorem{innercustomtrm}{\textbf{Theorem}}
\newenvironment{customtheorem}[1]
  {\renewcommand\theinnercustomtrm{#1}\innercustomtrm}
  {\endinnercustomtrm}
  
\newenvironment{boxfigH}[2]{%
     \begin{figure}[H]
     \newcommand{\FigCaption}{#1}
     \newcommand{\FigLabel}{#2}
     \begin{center}
       \begin{small}
         \begin{tabular}{@{}|@{~~}l@{~~}|@{}}
           \hline
           \rule[-1.5ex]{0pt}{1ex}\begin{minipage}[b]{.95\linewidth}
             \vspace{1ex}
             \smallskip
             }{%
           \end{minipage}\\
           \hline
         \end{tabular}
       \end{small}

     \end{center}
     \label{\FigLabel}
   \end{figure}
}

\begin{document}
\title{Definitions and Security of Quantum Electronic Voting}
\author{%
Myrto Arapinis\inst{1}  \and 
Elham Kashefi\inst{1,2}  \and
Nikolaos Lamprou\inst{1} \and 
Anna Pappa\inst{3,4}
}%
\institute{
School of Informatics, University of Edinburgh, UK \and
LIP6, University Pierre et Marie Curie, France\and
Department of Electrical Engineering and Computer Science, Technische Universität Berlin, Germany\and
Dahlem Center for Complex Quantum Systems, Freie Universität Berlin, Germany
}

\maketitle 
\begin{abstract}
Recent advances indicate that quantum computers will soon be reality. Motivated by this ever more realistic threat for existing classical cryptographic protocols, researchers have developed several schemes to resist ``quantum attacks''. In particular, for electronic voting, several e-voting schemes relying on properties of quantum mechanics have been proposed. However, each of these proposals comes with a different and often not well-articulated corruption model, has different  objectives, and is accompanied by security claims which are never formalized and are at best justified only against specific attacks. To address this, we propose the first formal security definitions for quantum e-voting protocols. With these at hand, we systematize and evaluate the security of previously-proposed quantum e-voting protocols; we examine the claims of these works concerning privacy, correctness and verifiability, and if they are correctly attributed to the proposed protocols. In all non-trivial cases, we identify specific quantum attacks that violate these properties. We argue that the cause of these failures lies in the absence of formal security models and references to the existing cryptographic literature.
\end{abstract}

\begin{keywords}
quantum electronic voting, quantum cryptography, attacks
\end{keywords}


\section{Introduction} \label{sec:introduction}
Voting is a fundamental procedure in democratic societies. With the technological advances of the computer era, voting could benefit to become more secure and efficient and as a result more democratic. \ignore{Indeed, compared to previous manual procedures, electronic voting systems can offer more efficient elections with higher voter participation, better accuracy, while also providing enhanced security guarantees, such as vote-privacy and voter-verification even in the face of untrusted election authorities.} For this reason, over the last two decades, several cryptographic protocols  for electronic voting were proposed and implemented~\cite{AB08,Scantegrity,JD05,KT15,PaV}. The security of all these systems relies on computational assumptions such as the hardness of integer factorization and the discrete logarithm problem. But, these are easy to solve with quantum computers using Shor's algorithm~\cite{Shor94}. Although not yet available, recent technological advances indicate that quantum computers will soon be built threatening existing cryptographic protocols. In this context, researchers have proposed to use quantum communication to implement primitives like key distribution, bit commitment and oblivious transfer. Unfortunately, perfect security without assumptions has proven to be challenging in the quantum setting \cite{LoChau98,Mayers97}, and the need to study different corruption models has emerged. This includes limiting the number of dishonest participants and introducing different non-colluding authorities.

More than a decade of studies on quantum electronic voting has resulted in several protocols that use the properties of quantum mechanical systems. However, all these new protocols are studied against different and not well-articulated corruption models, and claim security using ad-hoc proofs that are not formalized and backed only against limited classes of quantum attacks. In particular, none of the proposed schemes  provides rigorous definitions of privacy and verifiability, nor formal security proofs against specific, well-defined (quantum) attacker models. When it comes to electronic voting schemes, it is particularly hard to ensure that all the, somehow conflicting, properties hold; it is therefore important that these new quantum protocols be rigorously and mathematically studied and the necessary assumptions and limitations formally established.

This is precisely what we set to address in this paper. We first give formal definitions for verifiability and vote privacy in the quantum setting considering adaptive corruption. Subsequently, we systematize and assess the security of existing e-voting protocols based on quantum technology. We specifically examine the claims of each of these solutions concerning the above-mentioned well-defined properties. Unfortunately our analyses uncover vulnerabilities in all the proposed schemes. While some of them suffer from trivial attacks due to inconsistencies in the security definitions, the main contribution of the paper is to argue that sophisticated attacks can exist even in protocols that ``seem secure'' if the security is proven ad hoc, and not in a formal framework. We argue that the cause of these failures is the absence of an appropriate security framework in which to establish formal security proofs, which we have now introduced.

Therefore, this paper follows previous works~\cite{BW02,PC17,UD10} in their effort to highlight the importance of formally defining and proving security in the relatively new field of quantum cryptography. This also includes studying classical protocols that are secure against unbounded attackers \cite{BT07}, as well as ones based on problems believed to be hard even for quantum computers e.g. lattice-based ~\cite{CL16}. However, it is out of the scope of this study to review such classical protocols, as we are focusing on the possible contribution of quantum computers to the security of e-voting.

\noindent{\bf Contributions:} We propose the first formal definitions for vote privacy and universal verifiability in the quantum setting considering adaptive corruption, and show that none of the proposed quantum protocols so far satisfy them. To this end, we systematize the proposed quantum e-voting approaches according to key technical features. To our knowledge, our study covers all relevant research in the field, identifying four main families. Table~\ref{tab:summary} summarises our results.
\begin{table}\label{tab:summary}
  \centering
  \begin{tabular}{|c|c|c|c|} 
   \hline
 \diagbox[width=10em]{Protocols }{Security} & \makecell{Privacy} & \makecell{Correctness} & \makecell{Corruption}  \\
   \hline
   Dual basis measurement based protocols & \ding{53} & \textbf{?} & \text{$\epsilon$ fraction of voters}\\           
   \hline
    Travelling ballot based protocols & \ding{53} &  \ding{53} & \text{two voters}  \\
   \hline
   Distributed ballot based protocols &\textbf{?} & $\text{\ding{53}}^{*}$ & \text{$\epsilon$ fraction of voters}  \\
   \hline
   Quantum voting based on conjugate coding &\ding{53} &\textbf{?} & \text{election authority} \\
   \hline
  \end{tabular}
  \caption{\textbf{\ding{53}}: Insecure, \textbf{?}: Unexplored Area,  $\text{}^{*}$:Protocol runs less than $exp(\Omega(N))$ rounds.}
 \end{table}
\begin{itemize}
\item \emph{Two measurement bases protocols -} These protocols rely on two measurement bases to verify the correct distribution of an entangled state. We specifically prove that the probability that a number of corrupted states are not tested and used later in the protocol, is non-negligible, which leads to a violation of voters' privacy. Furthermore, even if the states are shared by a trusted authority, we show that privacy can still be violated in case of abort.
\item \emph{Traveling ballot protocols -} In these protocols the ``ballot box'' circulates among all voters who add their vote by applying a unitary to it. We show how colluding voters can break honest voters' vote privacy just by measuring the ballot box before and after the victim has cast their ballot. These protocols further suffer as we will see from double voting attacks, whereby a dishonest voter can simply apply multiple time the voting operator.
\item \emph{Distributed ballot protocols -} These schemes exploit properties of entangled states that allow voters to cast their votes by applying operations on parts of them. We present an attack that allows the adversary to double-vote and therefore change the outcome of the voting process with probability at least $0.25$, if the protocol runs fewer than exponentially many rounds in the number of voters. The intuition behind this attack is that an adversary does not need to find exactly how the ballots have been created in order to influence the outcome of the election; it suffices to find a specific relation between them from left-over voting ballots provided by the corrupted voters. 
\item \emph{Conjugate coding protocols -} These protocols exploit BB84 states adding some verification mechanism. The main issue with these schemes, as we show, is that ballots are malleable, allowing an attacker to modify the part of the ballot which encodes the candidate choice to their advantage.  
\end{itemize}


\section{Preliminaries}
\label{preliminaries}
We use the term quantum bit or qubit \cite{NI02} to denote the simplest quantum mechanical object we will use. We say that a qubit is in a pure state if it can be expressed as a linear combination of other pure states: \[\ket{x}= \alpha \ket{0} + \beta \ket{1}, \text{~where~}\ket{0}= \begin{bmatrix}
1 \\ 0
\end{bmatrix}, \ket{1}=\begin{bmatrix}
0 \\ 1
\end{bmatrix} \]
where $|\alpha|^2 +|\beta|^2 =1$ for any $\alpha, \beta \in \mathbb{C}$. The states $\ket{0}$ and $\ket{1}$ are called the computational basis vectors. Sometimes it is also helpful to think of a qubit as a vector in the two-dimensional Hilbert space $\mathcal{H}$.  If a qubit cannot be written in the above form, then we say it is in a mixed state. The generalization of a qubit to an $m$-dimensional quantum system is called \emph{qudit}: 
\[\ket{y}=\sum_{j=0}^{m-1} a_j \ket{j}, \text{~where~} \sum_{j=0}^{m-1} |a_j|^2 =1
\] 
Let's now suppose that we have two qubits; we can write the state vector as:

\[ \ket{\psi}=\sum\limits_{i,j\in \{0,1\}}^{} \alpha_{ij}\ket{ij}\]
where $\sum\limits_{i,j\in \{0,1\}}^{} |\alpha_{ij}|^2=1$. If the total state vector $\ket{\psi}$ cannot be written as a tensor product of two qubits (\emph{i.e.} $\ket{x_1}\otimes\ket{x_2}$), then we say that qubits $\ket{x_1}$ and $\ket{x_2}$ are entangled. An example of two-qubit entangled states, are the four \ignore{maximally entangled} \emph{Bell states}, which form a basis of the two-dimensional Hilbert space:
\begin{eqnarray*}
\ket{\Phi^\pm}&=\frac{1}{\sqrt{2}}(\ket{00}\pm\ket{11}), \ket{\Psi^\pm}&=\frac{1}{\sqrt{2}}(\ket{01}\pm\ket{10})
\end{eqnarray*}
A quantum system that is in one of the above states is also called an EPR pair~\ignore{, from the famous paradox} \cite{EPR35}. 
The way we obtain information about a quantum system is by performing a measurement using a family of linear operators $\{M_j\}$ acting on the state space of the system, where $j$ denotes the different outcomes of the measurement. It holds for the discrete and the continuous case respectively that: \[\sum_j M_j^{\dagger}M_j=\int M_j^{\dagger}M_j dj=\mathbb{I}\]  where $M_j^\dagger$ is the conjugate transpose of matrix $M_j$, and $\mathbb{I}$ the identity operator. For qudit $\ket{y}$, the probability that the measurement outcome is $\w$ is: $\Pr(\w)=\bra{y}M_{\w}^{\dagger}M_{\w}\ket{y}$ and in the continuous case $\Pr(\w\in[\w_{1}, \w_{2}])=\int_{\w_{1}}^{\w_{2}}\bra{y}M_j^{\dagger}M_j\ket{y} dj $.

For a single qubit  $\ket{x}=\alpha \ket{0} + \beta \ket{1}$, measurement in the computational basis will give outcome zero with probability $|\alpha|^2$ and outcome one with probability $|\beta|^2$. 
If our state is entangled, a partial measurement (i.e. a measurement in one of the entangled qudits), not only reveals information about the measured qudit, but possibly about the remaining state. For example, let us recall the Bell state $\ket{\Phi^{+}}$. A measurement of the first qubit in the computational basis will give measurement outcome 0 or 1 with equal probability and the remaining qubit will collapse to the state $\ket{0}$ or $\ket{1}$ respectively. 

In quantum cryptography, the correlations in the measurement outcomes of entangled states are frequently exploited. Another entangled state of interest used in Section \ref{DualBase}, gives measurement outcomes that sum up to zero when measured in the computational basis, and equal outcomes when measured in the Fourier basis (denoted by $\ket{}_{F}$). In the three-qubit case, the state is the following:
\begin{align*}
\ket{D}&=\dfrac{1}{\sqrt{3}}\big(\ket{0}_F\ket{0}_F\ket{0}_F+\ket{1}_F\ket{1}_F\ket{1}_F\big)
=\dfrac{1}{2} \big(\ket{000} +\ket{011}+\ket{101}+\ket{110}\big )
\end{align*}
Finally, the evolution of a closed quantum system can be described by the application of a unitary operator. Unitary operators are reversible and preserve the inner product.  Recall our first example, and let's say we would like to swap the amplitudes on state $\ket{x}$, then we can apply the operator $Z$ (known as NOT-gate):
\[Z\ket{x}=\beta \ket{0}+\alpha \ket{1},~~\text{where}~~Z= \begin{bmatrix}
0 & 1\\ 1 & 0
\end{bmatrix}
\]
The $Z$-gate is one of the Pauli operators, which together with $X$ and $Y$, as well as the identity operator $\mathbb{I}$, form a basis for the vector space of $2\times 2$ Hermitian matrices. These operators are unitaries, and as such preserve the inner product. 

A very important difference between quantum and classical information, is that there is no mechanism to create a copy of an unknown quantum state \cite{NI02}. This result, known as the \emph{no-cloning theorem}, is one of the fundamental advantages and at the same time limitations of quantum information. It becomes extremely relevant for cryptography, since brute-force types of attacks cannot be applied on quantum channels that carry unknown information. When verifying quantum resources however, it is necessary to apply a cut-and-choose technique in order to test that the received quantum states are correcting produced. The quantum source would therefore need to send exponentially many copies of the quantum state \cite{KE17}, in order for the verifier to measure most of them and deduce that with high probability, the remaining ones are correct.

\section{Definitions of secure quantum electronic voting}
\label{definitions}

Electronic voting protocols consist of election authorities, talliers, voters and bulletin boards \cite{AB08,JD05,KT15}. In this work, we will be dealing with protocols involving only one election authority $EA$ and/or one tallier $T$, as well as the voters $\mathcal{V}=\{V_k\}_{k=1}^N$. $EA$ sets the parameters of the protocol, $\mathcal{V}$ cast ballots and $T$ gathers the votes, computes and announces the election outcome. Informally, a voting protocol $\Pi$ has three distinct phases (\emph{setup}, \emph{casting}, and \emph{tally}) and running time proportional to a security parameter $\delta_{0}$. For formalising security we adopt the standard game-based security framework. The security of a protocol is captured by a game between a challenger $\mc{C}$ that models the honest parties, and a Quantum Probabilistic Polynomial Time adversary $\mc{A}$ that captures the corrupted parties. $\mc{A}$ can adaptively corrupt a fraction $\epsilon$ of the voters. We assume that the eligibility list is provided a priori in a trusted manner and that $\mc{A}$ chooses the honest votes, in order to provide stronger definitions \cite{JD05}. The order $\pe$ with which honest parties cast their ballots is initially unknown to $\mc{A}$ but might leak during the execution of the protocol, if for example anonymous channels are not used or the casting order is decided by the voters. Finally, we consider that the parties use quantum registers to communicate and store information (denoted by $\mathcal{B}$ and $\mathcal{X}$ respectively), to account for the case where the states are entangled between different parties.
\begin{description}
\item[\emph{Setup phase}:] $\mc{A}$  defines the voting choices of all voters. $\mc{C}$  and $\mc{A}$  generate the protocol parameters $\mathcal{X}$ according to $\Pi$.
\item[\emph{Casting phase}:] The protocol $\Pi$ specifies the algorithm $\mathsf{CastBallot}$ for generating and casting the ballots. $\mc{C}$ generates ballots according to the $\mathsf{CastBallot}$ algorithm on behalf of honest voters and $\mc{A}$  on behalf of the corrupted ones.
\item[\emph{Tally phase}:] The protocol $\Pi$ specifies the tallying algorithm $\mathsf{Tally}$. $\mc{C}$ computes the election result on behalf of the parties specified in $\Pi$ by running the $\mathsf{Tally}$ algorithm. If none of these parties is honest, $\mc{A}$ computes the tally instead.
\end{description}
Ideally, an e-voting protocol will satisfy at least the following properties (\cite{BD15,CV16,DSK06}; \textsf{Correctness}: compute the correct outcome if the adversary doesn't interfere, \textsf{Double voting}: allow voters to vote at most once, \textsf{Privacy}: keep the vote of a voter private, \textsf{Verifiability}: allow for verification of the results by voters and external auditors. We focus on privacy and verifiability type properties.

\subsubsection{Universal Verifiability -}
Our definition of \emph{universal verifiability} is similar to \cite{CV16} and is captured by the experiment $\mathbf{EXP}_{\msf{Qver}}^{\Pi}$. 

\begin{boxfigH}{The experiment $\mathbf{EXP}_{\msf{Qver}}^{\Pi}$
}{fig:Qcorr}
\underline{\emph{The experiment $\mathbf{EXP}_{\msf{Qver}}^{\Pi}(\mathcal{A},\epsilon,\delta_{0})$}}
\begin{itemize}
\item[--]\textbf{Setup phase:} $\mc{C}$ and $\mc{A}$ generate the protocol parameters in quantum register $\mc{X}$ as specified by $\Pi$ and the adversarial model. Furthermore, $\mc{A}$ chooses the votes for all voters $\{v_{k}\}_{V_{k} \in \mathcal{V}}$. 
\item[--] \textbf{Casting phase:} For each $k \in \{1,\ldots,|\mc{V}| \}$, 
\begin{itemize}
\item $\mc{A}$ chooses to corrupt $V_{\pe(k)}$ or not. If $\mc{A}$ decides to corrupt $V_{\pe(k)}$ is added to $\mc{V}_{\mc{A}}$.
\item If $V_{\pe(k)}\not\in\mc{V}_{\mc{A}}$, $\mc{C}$ generates the ballot $\{\mc{B}_{\pe(k)},\bot\} \leftarrow \mathsf{CastBallot}(v_{\pe(k)},\mc{X}_{\pe(k)},\mc{B},\delta_{0})$. If it is not $\bot$, $\mc{C}$ sends it to $\mc{A}$. If $\mc{C}$ receives $\mc{B}_{\pe(k)}$ back from $\mc{A}$, then $\mc{C}$ stores $\mc{B}_{\pe(k)}$ in $\mc{B}$, where $\{\mc{B}_{\pe(k)}$, $\mc{X}_{\pe(k)}\}$, and $\mc{B}$ are local and global quantum registers respectively. Note that when $\mc{A}$ receives $\mc{B}_{\pe(k)}$ from $\mc{C}$, it is possible to apply quantum operations on the register that are dependent on the specifications of $\Pi$. 
\item If $V_{\pe(k)}\in\mc{V}_{\mc{A}}$, then $\mc{A}$ creates a ballot $\mc{B}_{\pe(k)}$ and sends it to $\mc{C}$.
\end{itemize}
\item[--] \textbf{Tally phase:} If the tallier is corrupted, $\mc{A}$ outputs the election outcome $X$. Otherwise $\mc{C}$ computes $X \leftarrow \mathsf{Tally}(\mc{B},\mc{X}_{\mc{C}},\delta_{0}) $:
\begin{itemize}
\item If ($X \not= \bot \wedge \mathsf{Verify}^{\Pi}(X,\mc{B},\delta_{0})$) and $\msf{P}_{\msf{VCounted}}^{\Pi}(\{v_{k}\}_{V_{k} \not\in \mathcal{V}_{\mc{A}}},X) = 0$ or $\mathsf{Nballots}^{\Pi}(X)> |\mc{V}|$, then {\tt output 1}, else {\tt output 0}.
\end{itemize}
\end{itemize}
\end{boxfigH}

\vspace{-0.3in}

First, $\mc{A}$ defines how honest voters vote. Then, $\mc{C}$ and $\mc{A}$ generate the protocol parameters $\mc{X}$ according to $\Pi$ and the corruption model of $\mc{A}$.
In the \emph{casting phase}, $\mc{A}$ can choose to corrupt voters adaptively. For honest voters, $\mc{C}$ follows the $\mathsf{CastBallot}$  algorithm as specified by $\Pi$ to generate the ballot, and sends it to $\mc{A}$.
Depending on the protocol specification, $\mc{A}$ might then perform some allowed quantum operation on the received ballot (e.g. in the case where $\Pi $ uses quantum authenticated channels, the ballots cannot be modified by $\mc{A}$).
For corrupted voters, $\mc{A}$ casts the ballot on their behalf.  After all votes have been cast, the election outcome is computed; if the tallier is honest, $\mc{C}$ runs the algorithm $\mathsf{Tally}$ as specified by $\Pi$. $\mathbf{EXP}_{\msf{Qver}}^{\Pi}$ outputs $1$ if the election outcome is not $\bot$ and is accepted by $\mc{C}$, while either an honest vote has not be counted in the final outcome, or the number of cast votes exceeds the number of voters; Otherwise the experiment outputs $0$. To account for these events, we define three  predicates; $\mathsf{Verify}^{\Pi}$ which is the protocol-specific public test parties can run to verify the election, $\msf{P}_{\msf{VCounted}}^{\Pi}$ reveals if honest votes are discarded from or altered in the final outcome and $\mathsf{Nballots}^{\Pi}$ reveals the number of votes included the election result $X$.
If $\mc{A}$ deviates from the protocol specification, the predicate $\mathsf{Verify}^{\Pi}$ should return \texttt{false}. 

\begin{definition}
A quantum e-voting protocol $\Pi$ satisfies $\epsilon$-{\bf quantum verifiability} if for every QPPT $\mc{A}$ the probability of winning the experiment $\mathbf{EXP}_{\msf{Qver}}^{\Pi}(\mathcal{A},\epsilon,\delta_{0})$ is negligible with respect to $\delta_{0}$:

\[\Pr[1\leftarrow \mathbf{EXP}_{\msf{Qver}}^{\Pi}(\mathcal{A},\epsilon,\delta_{0})]=negl(\delta_{0}).\]
\end{definition}

\subsubsection{Vote privacy -} The experiment $\mathbf{EXP}_{\msf{Qpriv}}^{\Pi}$ captures vote privacy which ensures that the adversary $\mc{A}$ cannot link honest voters to their votes. 

\begin{boxfigH}{Experiment $\mathbf{EXP}_{\msf{Qpriv}}^{\Pi}$}{fig:Qpriv}\underline{\emph{The experiment $\mathbf{EXP}_{\msf{Qpriv}}^{\Pi}(\mathcal{A},\epsilon,\delta_{0})$}}
\begin{itemize}
\item[--] \textbf{Setup phase:} $\mc{A}$ chooses a permutation $\mathsf{F}^{\mc{A,V}}_l$, the set of candidates $\{v_{k}\}_{V_{k} \in \mathcal{V}}$ each voter will vote for. $\mc{C}$ and $\mc{A}$ generate the protocol parameters in quantum register $\mc{X}$ as specified by $\Pi$ and the adversarial model. $\mc{C}$ randomly picks bit $\beta \overset{\$}{\leftarrow} \{0,1\}$. 

\item[--] \textbf{Casting phase:} For each $k \in \{1,\ldots,|\mc{V}| \}$, 
\begin{itemize}
\item $\mc{A}$ chooses to corrupt $V_{\pe(k)}$ or not. If $\mc{A}$ decides to corrupt $V_{\pe(k)}$ is added to $\mc{V}_{\mc{A}}$.
\item If $V_{\pe(k)}\not\in \mc{V}_{\mc{A}}$, $\mc{C}$ generates the ballot $\{\mc{B}_{\pe(k)},\bot \}\leftarrow \mathsf{CastBallot}(\mathsf{F}^{\mc{A,V}}_l(v_{\pe(k)},\pe(k))^{\beta} \cdot v_{\pe(k)}^{1-\beta},\mc{X}_{\pe(k)},\mc{B},\delta_{0})$. If the generated ballot is not $\bot$, $\mc{C}$ sends it to $\mc{A}$. If $\mc{C}$ receives $\mc{B}_{\pe(k)}$ back from $\mc{A}$, then $\mc{C}$ stores $\mc{B}_{\pe(k)}$ to $\mc{B}$. Note that when $\mc{A}$ receives $\mc{B}_{\pe(k)}$ from $\mc{C}$, it is possible to apply quantum operations on the register that are dependent on the specifications of $\Pi$. 
\item If $V_{\pe(k)}\in \mc{V}_{\mc{A}}$, then $\mc{A}$ creates a ballot $\mc{B}_{\pe(k)}$ and sends it to $\mc{C}$.
\end{itemize}
\item[--]\textbf{Tally phase:} If  $\mathsf{F}^{\mc{A,V}}_l(\overline{\mc{V}_{\mc{A}}})=\mathsf{F}^{\mc{A},\overline{\mc{V}_{\mc{A}}}}_{l'}$, $\mc{C}$ announces the election outcome $X \leftarrow \mathsf{Tally}(\{\mc{B}_{\pe(k)}\}_{V_{k}\in \mc{V}}, \mc{X}_{\mc{C}}, \delta_{0})$ to $\mc{A}$. Else  {\tt output -1}.
\end{itemize}
$\mc{A}$ guesses bit $\beta^{*}$. If $\beta^{*}=\beta$ then {\tt output 1}, else {\tt output 0}.
\end{boxfigH}

$\mc{A}$ defines how honest voters vote and chooses a permutation  $\mathsf{F}^{\mc{A,V}}_l \in \mc{F}^{\mc{A,V}}$ over the voting choices of all voters in $\mc{V}$. After the parameters of the protocol $\mc{X}$ are generated, $\mc{C}$ chooses a random bit $\beta$ which defines two worlds; when $\beta = 0$, the honest voters vote as specified by $\mc{A}$, while when $\beta=1$, the honest voters swap their votes according to permutation $\mathsf{F}^{\mc{A,V}}_l$ again specified by $\mc{A}$. If the choices of the honest voters during the $\texttt{casting phase}$ are still a permutation of their initial choices the experiment proceeds to the next phase, else it outputs $-1$. In the tally phase, $\mc{C}$ computes the election outcome. Finally, $\mc{A}$ tries to guess if the honest voters controlled by $\mc{C}$ have permuted their votes ($\beta=1$) or not ($\beta=0$), by outputting guess bit $\beta^{*}$. If $\mc{A}$ guessed correctly $\mathbf{EXP}_{\msf{Qpriv}}^{\Pi}$ outputs $1$; otherwise $\mathbf{EXP}_{\msf{Qpriv}}^{\Pi}$ outputs $0$.

\begin{definition}
\label{Def:Qprivacy}
A quantum e-voting protocol $\Pi$ satisfies $\epsilon$-{\bf quantum privacy} if for every QPPT $\mc{A}$ the probability of winning the experiment $\mathbf{EXP}_{\msf{Qpriv}}^{\Pi}(\mathcal{A},\epsilon,\delta_{0})$ is negligibly close to $1/2$ with respect to $\delta_{0}$ under the condition event $\neg{\mathbf{False\_Attack}}$ happens, where $\mathbf{False\_Attack}=\{{\tt -1}\leftarrow\mathbf{EXP}_{\msf{Qpriv}}^{\Pi}(\mathcal{A},\epsilon,\delta_{0})\}$:
\[\Pr[1\leftarrow \mathbf{EXP}_{\msf{Qpriv}}^{\Pi}(\mathcal{A},\epsilon,\delta_{0})|\neg{\mathbf{False\_Attack}}]=1/2+negl(\delta_{0})\]
\end{definition}

The most established game-based definitions for privacy in the classical setting~\cite{BD15} assume two ballot boxes, where one holds the real tally and the other holds either the real or the fake tally. In the quantum case, the adaptation is not straightforward mainly because of the no-cloning theorem. The existence of two such boxes assumes that information is copyable, which is not the case with quantum information. Similarly, we can't assume that the experiment runs two times because $\mc{A}$ could correlate the two executions by entangling their parameters, something that a classical adversary cannot do. We address this difficulty by introducing quantum registers to capture the network activity and model the special handling of quantum information (e.g entangled states).
Moreover, the election result is produced on the actual ballots rather than the intended ones. With this, we capture a broader spectrum of attacks (e.g Helios replay attack), at the same time introduce trivial distinctions corresponding to false attacks. We tackle this by allowing the experiment to output {\tt -1} in such undesired cases which are mainly artifacts of the model. So an advantage of our privacy definition is that it allows the analysis of self-tallying type protocols in contrast with previous definitions of privacy~\cite{BD15}. In self-tallying elections the adversary is able to derive the election outcome on their own without the need of secret information. Therefore, the bulletin board (in our case the  register $\mc{B}$) must be consistent with the result without linking the identity and the vote of a voter. 

\paragraph{\bf Note -} Our definitions of verifiability and privacy capture both classical and quantum protocols. For the classical case, the quantum registers will be used for storing and communicating purely classical information. 
Devising our definitions for the quantum setting was not a trivial task as there are many aspects that are hard to define, like bulletin boards, and others that need to be introduced, like quantum registers potentially containing entangled quantum states. Moreover, our experiments capture protocols that use anonymous channels, by assuming that the casting order is unknown to $\mc{A}$, as well as self-tallying protocols. 

In the rest of the paper, we examine all existing proposals for quantum e-voting. For each of them, we identify attacks that violate the previously defined properties. Note that since the proposed protocols do not involve any verifiability mechanism, we need to define an experiment that involves an honest Tallier, and that captures security against double voting and vote deletion/alteration against malicious voters. We term this property \emph{integrity} and therefore need to consider experiment $\mathbf{EXP}_{\msf{Qint}}^{\Pi}$ which is the same as $\mathbf{EXP}_{\msf{Qver}}^{\Pi}$ but without the predicate $\mathsf{Verify}^{\Pi}$. The experiment $\mathbf{EXP}_{\msf{Qint}}^{\Pi}$ is detailed in Section~\ref{app:integrity-definition} of the Supplementary Material.


\section{Dual Basis Measurement Based Protocols}
\label{DualBase}

In this section we discuss protocols that use the dual basis measurement technique \cite{HW14,WQ16}, and use as a blank ballot an entangled state with an interesting property: when measured in the computational basis, the sum of the outcomes is equal to zero, while when measured in the Fourier basis, all the outcomes are equal. Both of these protocols use cut-and-choose techniques in order to verify that the state was distributed correctly. This means that a large amount of states are checked for correctness and a remaining few are kept at the end unmeasured, to proceed with the rest of the protocol. Although a cut-and-choose technique with just one verifying party is secure if the states that are sampled are exponentially many and the remaining ones are constant, it is not clear how this generalizes to multiple verifying parties. \ignore{In particular, if we consider an adversary who generates the blank ballots and can corrupt an arbitrary fraction of the voters, it turns out that the probability that some corrupted states are not tested, is non-negligible with respect to the security parameter of the protocol. }Specifically, we show that if the corrupted parties sample their states last, then the probability with which the corrupted states are not checked and remain after all the honest parties sample, is at least a constant with respect to the security parameter of the protocol.

\ignore{In addition, the protocol described in \cite{DS06} uses a hybrid of this technique and of the distributed ballot one \cite{BM11,HM06,VJ07}(which we present in more detail at section \ref{distributed}). Again, it suffers both from the privacy attack \ref{selftallying}  if we consider that the distribution of the state is not from a trusted party (as in \cite{HW14,WQ16}) and  from the double voting attack. This happens even if we use the elaborate technique in section \ref{distributed} proposed by \cite{HM06} to prevent double voting, because as we see, the protocol needs to run exponentially many times with respect to the number of voters in order to guaranty a negligible probability (with respect to the protocol's parameter) with which our adversary \NL{corrected} will succeed to double vote.}

\subsection{Protocol Specification}

\label{selftallying}
We will now present the self-tallying protocol of \cite{WQ16}, which is based on the classical protocol of \cite{KM02}. The voters $\{V_k\}_{k=1}^N$, without the presence of any trusted authority or tallier, need to verify that they share specific quantum states. At the end of the verification process, the voters share a classical matrix; every cast vote is equal to the sum of the elements of a row in the matrix. \\

\noindent{\bf Setup phase}
\begin{enumerate}
\item One of the voters, not necessarily trusted, prepares $N+N2^{\delta_0}$ states: 
\[\ket{D_1}=\frac{1}{\sqrt{m^{N-1}}}\sum _{\sum_{k=1}^N i_k= 0\mod{c}} \ket{i_1}\ket{i_2}\dots\ket{i_N}\]
where $m$ is the dimension of the qudits' Hilbert space, $c$ is the number of the possible candidates such that $m \ge c$ and $\delta_{0}$ the security parameter. The voter also shares $1+N2^{\delta_0}$ states of the form:  
\[\ket{D_2}=\frac{1}{\sqrt{N!}}\sum \limits_{(i_1,i_2,\dots,i_N)\in \mathcal{P}_{N}} \ket{i_1}\ket{i_2}\dots\ket{i_N}\]
where $\mathcal{P}_N$ is the set of all possible permutations with $N$ elements. Each $V_k$ receives the $k^{th}$ particle from each of the states. 
\item \label{2}The voters agree that the states they receive  are indeed $\ket{D_1},\ket{D_2}$ by using a cut-and choose technique. Specifically, voter $V_k$ chooses at random $2^{\delta_0}$ of the $\ket{D_1}$ states and asks the other voters to measure half of their particles in the computational and half in the Fourier basis. Whenever the chosen basis is the computational, the measurement results need to add up to $0$, while when the basis is the Fourier, then the measurement results are all the same. All voters simultaneously broadcast their results and if one of them notices a discrepancy, the protocol aborts. The states $\ket{D_2}$ are similarly checked.
\item \label{3} The voters are left to share $N$ copies of  $\ket{D_1}$ states and one $\ket{D_2}$ state. Each voter holds one qudit for each state. They now all measure their qudits in the computational basis. As a result, each $V_k$ holds a ``blank ballot" of dimension $N$ with the measurement outcomes corresponding to parts of $\ket{D_1}$ states:
\[B_{k}= [ \xi_{k}^{1}  \cdots \xi_{k}^{sk_{k}} \cdots \xi_{k}^{N}]^\intercal
 \] 
 and a unique index, $sk_{k}\in\{1,\dots,N\}$, from the measurement outcome of the qudit that belongs to $\ket{D_2}$. The set of all the blank ballots has the property $\sum_{k=1}^N\xi_k^j=0 \mod{c}$ for all $j=1,\dots,N$.
 \end{enumerate}
 
 \noindent{\bf Casting phase}
 \begin{enumerate}\setcounter{enumi}{3}
\item Based on $sk_{k}$, all voters add their vote, $v_{k} \in \mathbb{Z}_{c}$, to the corresponding row of their ``secret'' column. Specifically, $V_{k}$ applies $\xi_{k}^{sk_{k}} \rightarrow \xi_{k}^{sk_{k}}+ v_{k} $.
\item All voters simultaneously broadcast their columns, resulting in a public $N\times N$ table, whose $k$-th column encodes $V_{k}$'s candidate choice.
\[B=\begin{bmatrix}
 &  & \xi_{k}^{1} &  & \\
 & & \vdots & & \\
B_{1}^{v_{1}} & \cdots & \xi_{k}^{sk_{k}}+v_{k} & \cdots & B_{N}^{v_{N}} \\
 & & \vdots & & \\
 & & \xi_{k}^{N} &  &
\end{bmatrix} \]
\end{enumerate}

\noindent{\bf Tally phase}
\begin{enumerate}\setcounter{enumi}{5}
\item \label{check} Each $V_k$ verifies that their vote is counted by checking that the corresponding row of the matrix adds up to their vote. If this fails, the protocol aborts.
\item Each voter can tally the final outcome of the election by computing the sum of the elements of each row of the public $N \times N$ table. The resulting $N$ elements are the result of the election.
\end{enumerate}

\ignore{\begin{figure}[h]
\centering
\includegraphics[width=8cm]{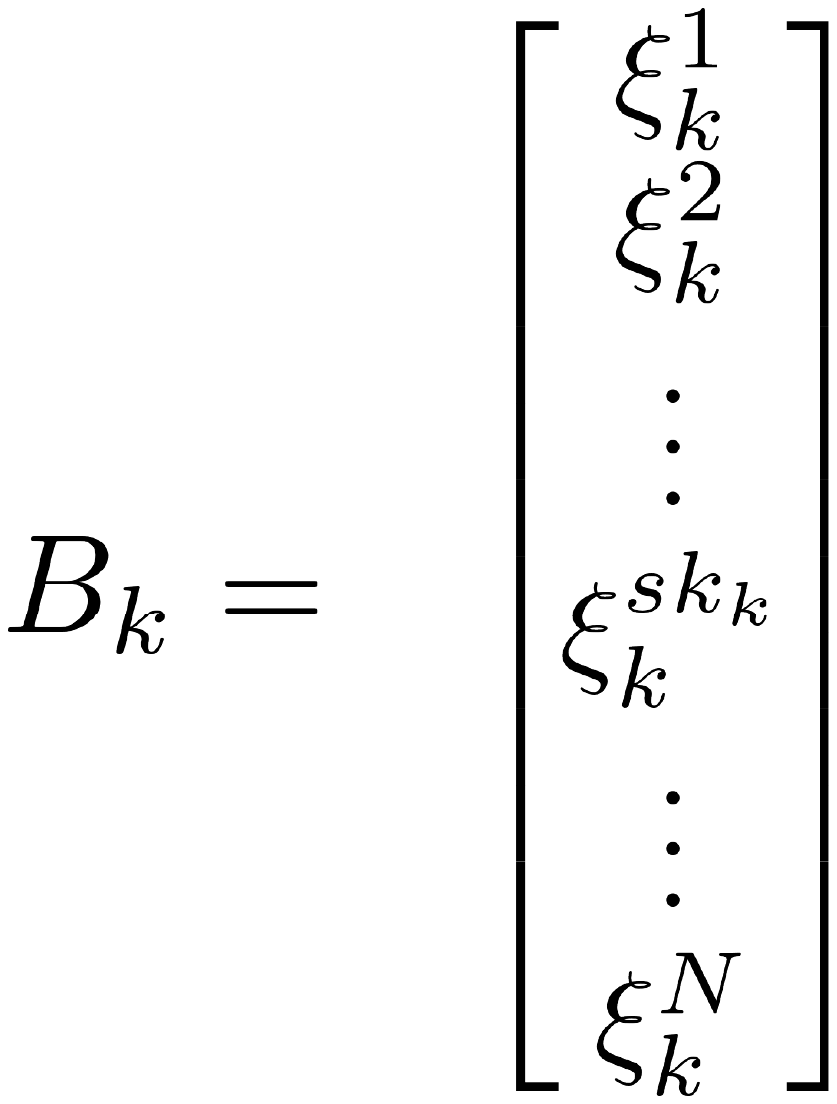}
\caption{The blank ballot of voter $V_{k}$ after it measures the states $\ket{D_{1}},\ket{D_{2}}$ in the computational basis.}
\label{example_Blank_CBallot}
\end{figure}
}
\ignore{\begin{figure}[h]
\centering
\includegraphics[scale=0.25]{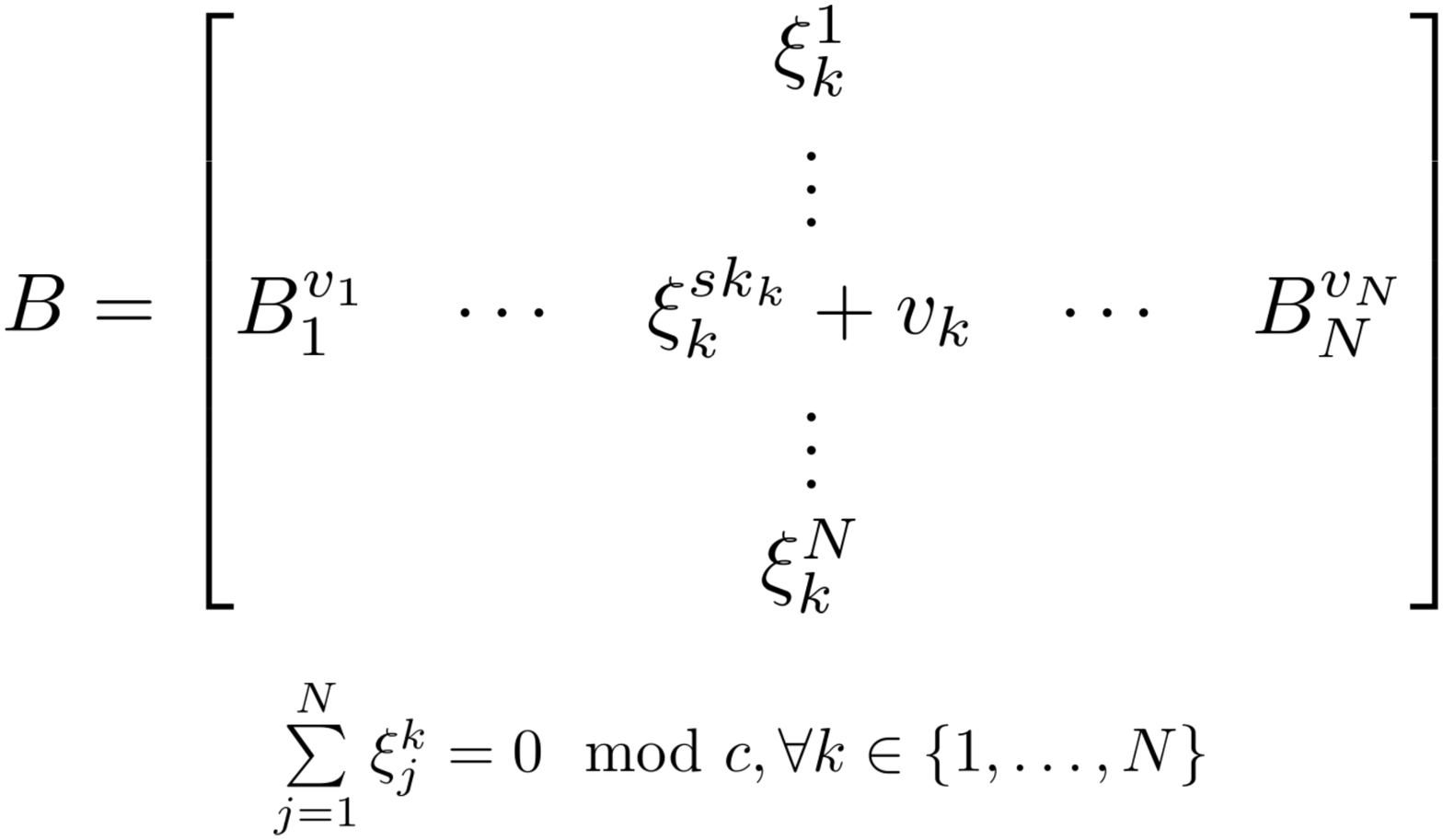}
\caption{The resulting $N \times N$ table after voters simultaneously broadcast their ballot column.}
\label{example_Consistent_View_Table}
\end{figure}}
\subsection{Vulnerabilities Of Dual Basis Measurement Protocols}

In this section we present an attack on the cut-and-choose technique of the protocol in the setup phase, that can be used to violate privacy. We consider a static adversary that corrupts $t$ voters, including the one that distributes the states. Suppose that the adversary corrupts $N$ out of $N+N 2^{\delta_0}$ states $\ket{D_{1}}$. We denote with \emph{Bad}, the event that all the corrupted voters choose last which states they want to test, and with \emph{Win}, the event that the $N$ corrupted states are not checked. We want to compute the probability that event \emph{Win} happens, given event \emph{Bad}, \emph{i.e.} the probability none of the $N$ corrupted states is checked by the honest voters, and therefore remain intact until the corrupted voters' turn. The corrupted voters will of course not sample any of the corrupted states and therefore the corrupted states will be accepted as valid.

The number of corrupted states that an honest voter will check, follows a mixture distribution with each mixture component being one of the hypergeometric distributions $\{\mathsf{HG}(L_{i_{k}},b_{i_{k}},2^{\delta_0}): 0 \leq b_{i_{k}} \leq N\}$ , where $L_{i_{k}}$ is the number of states left to sample from the previous voter and $b_{i_{k}}$ the number of the remaining corrupted states. We can therefore define the random variable $X_{i_{k}}$ that follows the above mixture distribution, where $i_1,\dots,i_{N-t}$ is a permutation of the honest voters' indices (by slightly abusing notation, we consider the first $N-t$ voters to be honest). The following lemma is proven by induction:

\begin{lemma}
Let $X_{i_{k}}$ be a random variable that follows the previous mixture distribution. It holds that: 
\[
\Pr[\sum\limits_{k=1}^{N-t} X_{i_{k}}=0]=\prod_{k=1}^{N-t} \Pr[X^{\ast}_{i_{k}}=0] \text{ where } X^{\ast}_{i_{k}} \sim \mathsf{HG}(L_{i_{k}},N,2^{\delta_0}).
\]
\end{lemma}
We are now ready to prove that with at least a constant probability, the corrupted states will remain intact until the end of the verification process.

\begin{proposition}
\label{main}
For $0< \varepsilon < 1$, let $t=\varepsilon N$ be the fraction of voters controlled by the adversary. It holds that :
\[\Pr[Win \mid Bad]> \Big(\frac{\varepsilon}{2}\Big)^{N}\]
\end{proposition}
\begin{proof} 
\begin{align*}
&\Pr[\emph{Win}\mid \emph{Bad}] =\Pr[\sum\limits_{k=1}^{N-t} X_{i_{k}}=0]= \prod_{k=1}^{N-t} P[X^{\ast}_{i_{k}}=0] \\
&= \prod_{k=0}^{N-t-1} \binom{N+N2^{\delta_0}-N-k 2^{\delta_0}}{2^{\delta_0}} / \binom{N+N2^{\delta_0} -k 2^{\delta_0}}{2^{\delta_0}}  \\  
&= \frac{(N+t2^{\delta_0} -N+1)\cdot \ldots \cdot (N+t2^{\delta_0})}{(N+N2^{\delta_0} -N+1)\cdot \ldots \cdot (N+N2^{\delta_0})} \\ 
& > \Big(\frac{t2^{ \delta_0}+1}{N+N2^{ \delta_{0}}}\Big)^N = \Big(\frac{t2^{ \delta_0}}{N+N2^{ \delta_{0}}} + \frac{1}{N+N2^{ \delta_{0}}}\Big)^N \\
&> \Big(\frac{t2^{ \delta_0}}{N+N2^{ \delta_{0}}}\Big)^N = \Big(\frac{\varepsilon}{2^{-\delta_{0}}+1}\Big)^N > \Big(\frac{\varepsilon}{2}\Big)^{N}
\end{align*}
\end{proof}
The question now is with what probability event \emph{Bad} occurs, \emph{i.e} how likely is the fact that voters controlled by the adversary are asked to sample last? The answer is irrelevant, because this probability depends on $N$ and $t$, and are both independent of $\delta_0$. As a result, 
\begin{align*}
&\Pr[\emph{Win}]>\Pr[\emph{Win}\mid \emph{Bad}]
\Pr[\emph{Bad}]=(\varepsilon/2)^N f(N,t)
\end{align*}
where $f(N,t)$ is a constant function with respect to the security parameter $\delta_{0}$, 
 making $\Pr[\emph{Win}]$ non-negligible in $\delta_{0}$. As a matter of fact, a static adversary will corrupt the voters that maximize $\Pr[\emph{Bad}]$. Therefore, we can assume that the honest voters sample the states at random, in order to not favor sets of corrupted voters. Now let us examine how this affects the privacy of the scheme.

\begin{theorem}
\label{main2}
Let $\Pi(N,t,\delta_0)$ be an execution of the self-tallying protocol with $N$ voters, $t$ of them corrupted, and $\delta_0$ the security parameter. We can construct an adversary $\mathcal{A}$, which with non-negligible probability in $\delta_0$ violates privacy. 
\end{theorem}
\begin{proof}
Let $\mathcal{C_{A}}$ be the set of indices of the corrupted voters with $|\mathcal{C_{A}}|=t$. Suppose the voter distributing the states is also corrupted, and prepares $1+N2^{\delta_{0}}$ states of the form of $\ket{D_{2}}$, $N2^{\delta_{0}}$ states of the form $\ket{D_{1}}$ and $N$ states of the form:
\[\ket{D_{\text{Corrupt}}}=\ket{\xi_{1}}\otimes \ldots\otimes \ket{\xi_{N}}\]
where $\xi_{k}\in_{R} \{0,\ldots,c-1\}$ for all $k\in \{2,\ldots,N\}$, $\xi_{1} \in  \{0,\ldots,c-1\}$ such that\footnote{\label{UaR}$\in_{R}$ denotes that the element is chosen uniformly at random from a specific domain.}:
\[\xi_{1}+ \ldots +\xi_{N}=0 \mod{c}\] 

From Proposition \ref{main} and the previous observations we know that the probability that states $\ket{D_{\text{Corrupt}}}$ remain intact after the verification procedure in step \ref{2} (\emph{i.e.} event $\emph{Win}$), happens with non-negligible probability in the security parameter $\delta_{0}$. Therefore, with non-negligible probability, the remaining states in step \ref{3} are: one of the form $\ket{D_2}$ and $N$ of the form $\ket{D_{\text{Corrupt}}}$. All honest voters $V_{k}$ measure their qudits  in the computational basis and end up with a secret number $sk_k$ (from measuring the corresponding part of $\ket{D_{2}}$) and a column 
\[B_{k}=[\xi_{k}^{1}~ \dots ~\xi_{k}^{sk_{k}}~ \dots~ \xi_{k}^{N}]^\intercal
\] (from measuring states $\ket{D_{\text{Corrupt}}}$), that is known to the adversary. Now all voters apply their vote $v_{k}$ to the $B_{k}$ according to $sk_k$. As a result: 
\[B_{k}^{v_{k}}=[\xi_{k}^{1} ~\dots ~\xi_{k}^{sk_{k}} + v_{k}~ \dots~ \xi_{k}^{N}]^\intercal  
\]
At this point all voters simultaneously broadcast their $B_{k}^{v_{k}}$, as the protocol specifies, and end up with the matrix $B=(B_{1}^{v_{1}}\ldots B_{N}^{v_{N}})$. Each $V_{k}, k\not\in \mathcal{C_{A}}$ checks that \[\sum\limits_{j=1}^{N} B[sk_{k},j]=v_{k} \mod{c}\]
which happens with probability $1$ from the description of the attack in the previous steps. As a result, each voter accepts the election result. The adversary knowing both the pre-vote matrix and the post-vote matrix can therefore extract the vote of all honest voters.
\end{proof}
A similar attack can be mounted if the adversary instead of corrupting $N$ out of $N+N2^{\delta_{0}}$ $\ket{D_{1}}$ states, corrupts just $1$ of the $\ket{D_{2}}$ states. The attack is similar to the one mentioned above but in this case the adversary knows the row in which each voter voted instead of the pre-vote matrix. Moreover, the probability of theorem \ref{main} is improved from $(\varepsilon/2)^N$ to $\varepsilon/2$ (the proof works in a similar way).

Based on the previous observations we can construct an adversary that violates $\epsilon$-{\bf quantum privacy} for any $0<\epsilon<1$ and any non-trivial permutation.

\begin{theorem}
The quantum self tallying protocol $\Pi(N,N/\epsilon,\delta_{0})$ doesn't satisfy the $\epsilon$-{\bf quantum privacy} property for any $0<\epsilon<1$.
\end{theorem}

\begin{proof}
(sketch) First $\mc{A}$ picks a non-trivial permutation $\mathsf{F}^{\mc{A}}$. It is easy to see that $\mc{A}$ can corrupt the quantum states in experiment $\mathbf{EXP}_{\msf{Qpriv}}^{\Pi}$ and $\mc{C}$ accepts with probability at least $\alpha(\delta_{0})$ the corrupted parameters, where $\alpha()$ a non negligible function, based on theorem \ref{main}. Next, $\mc{A}$ by reading all the quantum registers $\mc{B}_{j}$ one by one, can find how each voter has voted individually as in theorem \ref{main2}. As a result, $\mc{A}$ can find out if the honest voters have permute their votes or not and guess the challenge bit $\beta$ with probability at least $1/2+ \alpha(\delta_{0})$. Note that the probability for that $A$ the $\mathbf{EXP}_{\msf{Qpriv}}^{\Pi}$ to output $-1$, is $0$, because $\mc{A}$ can choose which voters to corrupt such that both the cut-and-choose attack in the \texttt{Setup phase} was successful and the condition $\mathsf{F}^{\mc{A,V}}_l(\overline{\mc{V}_{\mc{A}}})=\mathsf{F}^{\mc{A},\overline{\mc{V}_{\mc{A}}}}_{l'}$ is satisfied.
\end{proof}

So far we have seen how voters' privacy can be violated if an adversary distributes the quantum states in the protocol. However, even if the sharing of the states is done honestly by a trusted authority, still an adversary $\mc{A}$ can violate the privacy of a voter. This is done by replacing one element in a column of one of the players controlled by $\mc{A}$ with a random number. As a result, in step $\ref{check}$, the honest voter whose row doesn't pass the test, will abort the protocol by broadcasting it.  $\mc{A}$ will therefore know the identity of the voter aborting and their corresponding vote, since it knows the matrix before the modification of the column element. Similarly, in experiment $\mathbf{EXP}_{\msf{Qpriv}}^{\Pi}$ the adversary can find out if the voters permute their vote or not.

A possible solution might be to use a classical anonymous broadcast channel, so that the voters can anonymously broadcast abort if they detect any misbehaviour at step $\ref{check}$. However, this might open a path to other types of attacks, like denial-of-service, and requires further study in order to be a viable solution.


\section{Traveling Ballot Based Protocols}
\label{traveling}
In this section we discuss the traveling ballot family of protocols for referendum type elections. Here, $\A$ also plays the role of $EA$, as it sets up the parameters of the protocol in addition to producing the election result. Specifically, it prepares two entangled qudits, and sends one of them (the \emph{ballot qudit}) to travel from voter to voter. When the voters receive the ballot qudit, they apply some unitary operation according to their vote and forward the qudit to the next voter. When all voters have voted, the ballot qudit is sent back to $\A$ who measures the whole state to compute the result of the referendum. The first quantum scheme in this category was introduced by Vaccaro \etal~\cite{VJ07} and later improved \cite{BM11,HM06,LY08}. 

\subsection{Protocol Specification}
Here we present the travelling ballot protocol of \cite{HM06}; an alternative form \cite{VJ07} encodes the vote in a phase factor rather than in the qudit itself.
\begin{description}
\item[Set up phase] $\A$ prepares the state $\ket{\Omega_{0}}=\frac{1}{\sqrt{N}}\sum\limits_{j=0}^{N-1}\ket{j}_{V}\ket{j}_{\A}$, keeps the second qudit and passes the first (the ballot qudit) to voter $V_{1}$.
\item[Casting phase]
For $k=1,\dots,N$, $V_{k}$ receives the ballot qudit and applies the unitary $U^{v_k}=\sum\limits_{j=0}^{N-1}\ket{j+1}\bra{j}$, where $v_k=1$ signifies a ``yes vote and $v_k=0$ a ``no'' vote (i.e. applying the identity operator). Then, $V_k$ forwards the ballot qudit to the next voter $V_{k+1}$ and $V_{N}$ to to $\A$. 
\item[Tally phase] The global state held by $\A$ after all voters have voted, is: \[\ket{\Omega_{N}}=\frac{1}{\sqrt{N}}\sum\limits_{j=0}^{N-1}\ket{j+m}_{V}\ket{j}_{\A} \]
where $m$ is the number of "yes" votes. $\A$ measures the two qudits in the computational basis, subtracts the two results and obtains the outcome $m$. 

\end{description}

\subsection{Vulnerabilities Of Traveling Ballot Based Protocols}
The first obvious weakness of this type of protocols is that they are subject to double voting. A corrupted voter can apply the ``yes'' unitary operation many times without being detected (this issue is addressed in the next session, where we study the distributed ballot voting schemes). As a result we can easily construct an adversary $\mc{A}$ that wins in the integrity experiment described in the Appendix (Figure \ref{fig:Qint}) with probability $1$.
Furthermore, these protocols are subject to privacy attacks, when several voters are colluding. In what follows, we describe such an attack on privacy, in the case of two colluding voters. Figure~\ref{attack_TB} depicts this attack.

Let us assume that the adversary corrupts voters $V_{k-1}$ and $V_{k+1}$ for any $k$. Upon receipt of the ballot qudit, instead of applying the appropriate unitary, $V_{k-1}$ performs a measurement on the traveling ballot in the computational basis. As a result the global state becomes $\ket{\Omega_{k-1}}=\ket{h+m}_{V}\otimes \ket{h}_{\A}$, 
where $\ket{h+m}_{V}$ is one of the possible eigenstates of the observable $O=\sum\limits_{j=0}^{N-1} \ket{j}\bra{j}$, and $m$ is the number of ``yes'' votes cast by the voters $V_{1},\ldots, V_{k-2}$ (note that $V_{k-1}$ does not get any other information about the votes of the previous voters, except number $h+m$). Then $V_{k-1}$ passes the ballot qudit $\ket{h+m}_{V}$ to $V_{k}$, who applies the respective unitary for voting ``yes'' or ``no''. As a result the ballot qudit is in the state $\ket{h+m+v_{k}}_{V} $. Next, the ballot qudit is forwarded to the corrupted voter $V_{k+1}$, who measures it again in the computational basis and gets the result $h+m+v_{k}$. $\mc{A}$ can now infer vote $v_{k}$ from the two measurement results and figure out how $V_{k+1}$ voted. 
Similarly, $\mc{A}$ can guess the correct bit $\beta$ in $\mathbf{EXP}_{\msf{Qpriv}}^{\Pi}$ with probability $1$ by measuring the quantum registers $\mc{B}_{\pe(j)-1}$ and $\mc{B}_{\pe(j)+1}$, where $V_{\pe(j)}$ an honest party ($\mc{A}$ might need to corrupt two more voters such that $\mathbf{EXP}_{\msf{Qpriv}}^{\Pi}$ will not output {\tt -1}). The same attack can also be applied in the case where there are many voters between the two corrupted parties. In this case the adversary can't learn the individual votes but only the total votes.
\begin{figure}
\centering
\includegraphics[scale=0.95]{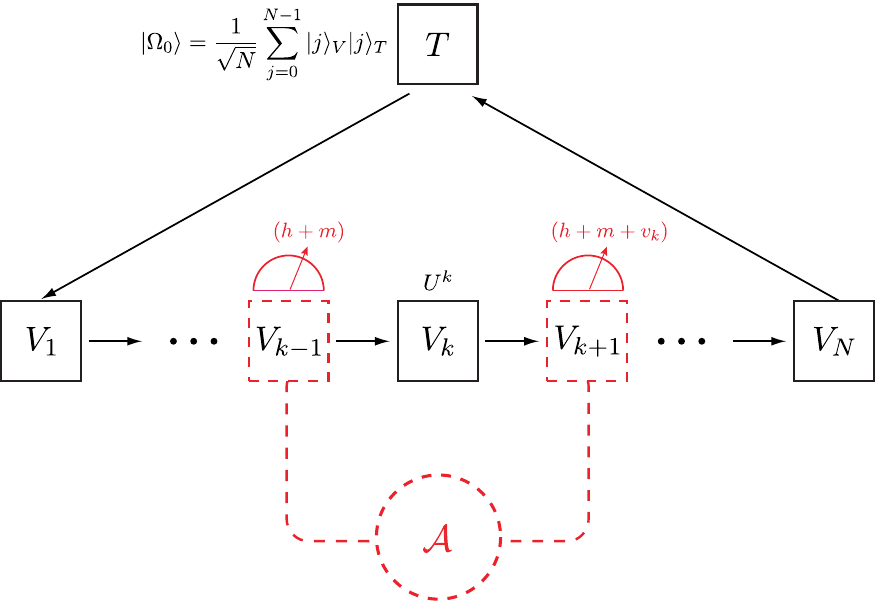}
\caption{$\mc{A}$ corrupts voters $V_{k-1}$,$V_{k+1}$ and learns how voter $V_{k}$ voted with probability $1$.}
\label{attack_TB}
\end{figure}
One suggestion presented in \cite{VJ07} is to allow $\A$ to perform extra measurements to detect a malicious action during the protocol's execution. However, this only identifies an attack and does not prevent the adversary from learning some of the votes, as described above. Furthermore, the probability of detecting a deviation from the protocol is constant and as such does not depend on the security parameter and does not lead to a substantial improvement of security.
It should also be noted that verifiability of the election result is not addressed in any of these works, since $\A$ is assumed to generate the initial state honestly. In the case where $\A$ is corrupted, privacy is trivially violated.

All traveling ballot protocols proposed (\cite{VJ07,BM11,HM06,LY08}) suffer from the above privacy attack. Next, we discuss how this issue has been addressed by revisiting the structure of the protocols. Unfortunately, as we will see, new issues arise.


\section{Distributed Ballot Based Protocols}\label{distributed}
Here we describe the family of quantum distributed ballot protocols \cite{BM11,HM06,VJ07}. In these schemes, $\A$ prepares and distributes to each voter a blank ballot, and gathers it back after all voters have cast their vote in order to compute the final outcome. This type of protocols give strong guarantees for privacy against other voters but not against a malicious $\A$ which is trusted to prepare correctly specific states. So it is not hard to see that if the states are not the correct ones, then the privacy of a voter can be violated. 

A first attempt presented in \cite{VJ07} suffers from double voting similarly to the discussion in the previous section. The same problem also appears in \cite{DS06}. Later works \cite{BM11,HM06} address this issue with a very elaborate countermeasure. The intuition behind the proposed technique is that $\A$ chooses a secret number $\delta$ according to which it prepares two different quantum states: the ``yes'' and the ``no'' states. This $ \delta $ value is hard to predict due to the non-orthogonality of the shared states and the no-cloning theorem. The authors suggest that many rounds of the protocol be executed. As a result, any attempt of the adversary to learn $\delta$ gives rise to a different result in each round. However, the number of required rounds, as well as a rigorous proof are not presented in the study. 

More importantly, a careful analysis reveals that the proposed solution is still vulnerable to double voting. As we will see, an adversary can mount what we call a $d$-transfer attack, and transfer $d$ votes for one option of the referendum election to the other. To achieve this attack, the adversary does not need to find the exact value of $\delta$ (as the authors believed), but knowing the difference of the angles used to create the ``yes'' and ``no'' states suffices. We construct a polynomial quantum adversary that performs the $d$-transfer attack with probability at least $0.25$, if the number of rounds is smaller than exponential in the number of voters. As a result this makes the protocol practically unrealistic for large scale elections.

\subsection{Protocol Specification}
\label{DistrDisc}
\noindent We first present the protocol from \cite{BM11,HM06}:

\noindent{\bf Setup phase}
\begin{enumerate}
\item $\A$ prepares an $N$-qudit ballot state: $\ket{\Phi}=\frac{1}{\sqrt{D}}\sum_{j=0}^{D-1}\ket{j}^{\otimes N}$, 
where the states $\ket{j}, j =0,...,D-1$, form an orthonormal basis for the $D$-dimensional Hilbert space, and $D>N$. The $k$-th qudit of $\ket{\Phi}$ corresponds to $V_k$'s blank ballot.
\item $\A$ sends to $V_{k}$ the corresponding blank ballot together with two option qudits, one for the ``yes'' and one for the ``no'' option:
\begin{eqnarray*}
\text{yes:}&~~\ket{\psi(\theta_y)}=\frac{1}{\sqrt{D}}\sum_{j=0}^{D-1}e^{ij\theta_y}\ket{j}, \text{no:}&~~\ket{\psi(\theta_n)}=\frac{1}{\sqrt{D}}\sum_{j=0}^{D-1}e^{ij\theta_n}\ket{j}
\end{eqnarray*}
For $v\in\{y,n\}$ we have $\theta_v=(2\pi l_v/D)+\delta$, where $l_v\in\{0,\dots, D-1\}$ and $\delta\in [0,2\pi/D)$. Values $l_{y}$ and $\delta$ are chosen uniformly at random from their domain and $l_{n}$ is chosen such that $N(l_{y}-l_{n} \mod{D})<D$ . These values are known only to $\A$.
\end{enumerate}

\noindent{\bf Casting phase}
\begin{enumerate}\setcounter{enumi}{2}
\item Each $V_{k}$ decides on ``yes'' or ``no'' by appending the corresponding option qudit to the blank ballot and performing a 2-qudit measurement $R=\sum_{r=0}^{D-1} rP_r$, where:
\[
P_r=\sum_{j=0}^{D-1}\ketbra{j+r}{j+r}\otimes\ketbra{j}{j}
\]
According to the result $r_k$, $V_k$ performs a unitary correction $U_{r_k}=\mathbb{I}\otimes \sum_{j=0}^{D-1}\ketbra{j+r_k}{j}$ and sends the 2-qudits ballot along with $r_k$ back to $\A$.
\end{enumerate}

\noindent{\bf Tally phase}
\begin{enumerate}\setcounter{enumi}{3}
\item The global state of the system (up to normalization) is:
\begin{align*}
&\ \frac{1}{\sqrt{D}} \sum\limits_{j=0}^{D-1} \prod\limits_{k=1}^{N}\alpha_{j,r_k} \ket{j}^{\otimes 2N}
\end{align*}
where, \[\alpha_{j,r_{k}}=\begin{cases}
  e^{i(D+j-r_{k})\theta_v^k} ,& 0 \leq j \leq r_{k}-1 \\ e^{i(j-r_{k})\theta_v^k} & r_{k} \leq j \leq D-1
\end{cases}\]
\item For every $k$, using the announced results $r_{k}$, T applies the unitary operator:
\[W_{k}=\sum_{j=0}^{r_{k}-1}e^{-iD\delta} \ket{j}\bra{j}+\sum_{j=r_{k}}^{D-1}\ket{j}\bra{j} \]
on one of the qudits in the global state (it is not important on which one, since changes to the phase factor of a qudit that is part of a bigger entangled state take effect globally). Now $\A$ has the state:
\[
\ket{\Omega_m}=\frac{1}{\sqrt{D}}\sum_{j=0}^{D-1}e^{ij(m\theta_y+(N-m)\theta_n)}\ket{j}^{\otimes 2N}
\]

where $m$ is the number of ``yes'' votes.
\item By applying the unitary operator $\sum_{j=0}^{D-1}e^{-ijN\theta_n}\ket{j}\bra{j}$ on one of the qudits and setting $q=m(l_y-l_n)$, we have:
\[
\ket{\Omega_q}=\frac{1}{\sqrt{D}}\sum_{j=0}^{D-1}e^{2\pi ijq/D}\ket{j}^{\otimes 2N}
\]
\end{enumerate}
We note here that $q$ must be between $0$ and $D-1$, so that the different outcomes be distinguishable. Now with the corresponding measurement $\A$ can retrieve $q$. Since $\A$ knows values $l_{y}$ and $l_{n}$, it can derive the number $m$ of "yes" votes. Note that if a voter does not send back a valid ballot, the protocol execution aborts.

\subsection{Vulnerabilities Of Distributed Ballot Protocols}
In this section, we show how the adversary can perform the $d$-transfer attack in favor of the ``yes'' outcome. We proceed as follows. We first show that this is possible if the adversary knows the difference $l_{y}-l_{n}$. We then show how the adversary can find out this value, and conclude the section with the probabilistic analysis of our attack which establishes that it can be performed with overwhelming probability in the number of voters.
\medskip

\noindent {\bf The d-transfer attack:\label{attackkvote}} Given the difference $l_{y}-l_{n}$, a dishonest voter can violate the no-double-voting. From the definition of $l_{y}$ and $l_{n}$ it holds that:
\begin{equation}
\label{dif}
  2\pi (l_{y}-l_{n})/D =\theta_{y}-\theta_{n}
\end{equation}
If a corrupted voter (e.g. $V_{1}$) knows $l_{y}-l_{n}$, then they proceed as follows (w.l.o.g. we assume that they want to increase the number of ``yes'' votes by $d$):
\begin{enumerate}
\item $V_{1}$ applies the unitary operator: $C_d=\sum\limits_{j=0}^{D-1} e^{ijd(\theta_y-\theta_n)}\ket{j}\bra{j}$ to the received option qudit $\ket{\psi(\theta_{y})}$. As a result, the state becomes: 
$$C_d \ket{\psi(\theta_y)}=\frac{1}{\sqrt{D}}\sum\limits_{j=0}^{D-1}e^{ ijd(\theta_y-\theta_n)} e^{i j \theta_{y}} \ket{j}$$

\item $V_{1}$ now performs the 2-qudit measurement specified in the Casting phase of the protocol and obtains the outcome $r_{1}$.

\item $V_{1}$ performs the unitary correction $ U_{r_{1}}$. For $\tilde{\theta}=d(\theta_y-\theta_n)+\theta_y$, the global state now is:
\begin{align*}
&U_{r_{1}} P_{r_{1}}\big(\ket{\Phi}\otimes C_d\ket{\psi(\theta_y)}\big)= \\ &\frac{1}{\sqrt{D}}\Big[\sum_{j=0}^{r_{1}-1}e^{i(D+j-r_{1})\tilde{\theta}}\ket{j}^{\otimes{N+1}}+\sum_{j=r_{1}}^{D-1}e^{i(j-r_{1})\tilde{\theta}}\ket{j}^{\otimes{N+1}}\Big]
\end{align*}
\item Before sending the two qudit ballot and the value $r_{1}$  to $\A$, $V_{1}$ performs the following operation to the option qudit: 
\[\mathbf{Correct}_{r_1}= \begin{cases}
e^{-iD d(\theta_y-\theta_n)} \ket{j} \bra{j}, & 0 \leq j \leq r_1-1 \\ \ket{j} \bra{j} & r_1 \leq j \leq D-1

\end{cases}\]

\item After all voters have cast their ballots to $\A$, the global state of the system (up to normalization) is:
\begin{align*}
&\ \frac{1}{\sqrt{D}} (\sum\limits_{j=0}^{r_{1}-1}e^{i(j-r_{1})d( \theta_{y}-\theta_{n})} e^{i(D+j-r_1)\theta_{y}} \prod\limits_{k=2}^{N}\alpha_{j,r_k} \ket{j}^{\otimes 2N} \\ &\ +\sum\limits_{j=r_1}^{D-1}e^{i(j-r_1) d(\theta_{y}-\theta_{n})} e^{i(j-r_1)\theta_{y}} \prod\limits_{k=2}^{N}\alpha_{j,r_{k}} \ket{j}^{\otimes 2N})
\end{align*}
where, \[\alpha_{j,r_{k}}=\begin{cases}
  e^{i(D+j-r_{k})\theta_v^k} ,& 0 \leq j \leq r_{k}-1 \\ e^{i(j-r_{k})\theta_v^k} & r_{k} \leq j \leq D-1

\end{cases}\]
and $\theta_v^k$ describes the vote of voter $V_{k}$, where $v\in\{y,n\}$. $\A$ just follows the protocol specification. It applies some corrections on the state given the announced results $r_{k}$ and finally the state becomes: \[\frac{1}{\sqrt{D}}\sum\limits_{j=0}^{D-1}e^{i(j-r_{1}) d (\theta_{y}-\theta_{n})} e^{i(j-r_1)\theta_{y}} \cdot \ldots \cdot e^{i(j-r_n)\theta_v^n} \ket{j}^{\otimes2N}\]
which under a global phase factor is equivalent to: \[\frac{1}{\sqrt{D}}\sum\limits_{j=0}^{D-1}e^{ij d(\theta_{y}-\theta_{n})} e^{i j(m\theta_{y}+(N-m)\theta_{n})}  \ket{j}^{\otimes 2N} \]

\item $\A$ removes the unwanted factor $e^{ijN\theta_{n}}$ as prescribed by the protocol, and the final state is:
\begin{align*}
\ket{\Omega_{m+d}} &= \frac{1}{\sqrt{D}}\sum\limits_{j=0}^{D-1}e^{ij d(\theta_{y}-\theta_{n})} e^{i jm(\theta_{y}-\theta_{n})}  \ket{j}^{\otimes 2N } \\  
&= \frac{1}{\sqrt{D}}\sum\limits_{j=0}^{D-1}e^{2\pi ij(m+d)(l_y-l_n)/D}  \ket{j}^{\otimes 2N }
\end{align*}

\item After measuring the state, the result is $m+d$ instead of $m$.

\end{enumerate}
\medskip

\noindent {\bf Finding the difference between $l_{y}$ and $l_{n}$:} What remains in order to complete our attack is to find the difference $l_{y}-l_{n}$. We now show how an adversary can learn this difference with  overwhelming probability in $N$. We assume that the adversary controls a fraction $\varepsilon$ of the voters ($0<\varepsilon<1$), who are (all but one) instructed  to vote half the times "yes" and the other half "no". Instead of destroying the remaining option qudits (exactly $\varepsilon N/2$ "yes" and $\varepsilon N/2$ "no" votes), the adversary keeps them to run Algorithm \ref{adval}.
\begin{algorithm}
\algsetup{linenosize=\tiny}
\scriptsize
\caption{Adversary's algorithm}\label{adval}
\begin{algorithmic}[1]

\REQUIRE $D,\ket{\psi(\theta_v)}_1,\cdots,\ket{\psi(\theta_v)}_{\varepsilon N/2}$ 
\ENSURE $\tilde{l}\in\{0,\dots,D-1\}$ \\
\STATE{$\mathtt{Record}=[0,\dots,0]\in \mathbb{N}^{1\times D}$;\COMMENT{{\tiny \textcolor{gray}{This vector shows us how many values are observed in each interval}}}} 
\STATE{$\mathtt{Solution}=["Null","Null"]\in \mathbb{N}^{1\times 2}$;}\STATE{$i,l,m=0$;}

\WHILE{$i \leq \varepsilon N/2 $} \STATE{Measure $\ket{\psi(\theta_v)}_i$ by using POVM operator $E(\theta)$ from Eq.(\ref{POVM}), the result is $y_i$;} \STATE{Find the interval for which $\frac{2\pi j}{D} \leq y_i \leq \frac{2\pi (j+1)}{D}$;} \STATE{ $\mathtt{Record}[j]=$++;} \STATE{$ i$++;}
\ENDWHILE
\WHILE{$l < D $} 
\IF{$\mathtt{Record}[l]\geq 40 \% (\varepsilon N/2)$} \STATE{$\mathtt{Solution}[m]=l$;}\STATE{ $m++;$} \ENDIF
\STATE{$l++$; } \ENDWHILE \IF{$\mathtt{Solution}==[0,D-1]$}\STATE{$\mathtt{Solution}=[\mathtt{Solution}[1],\mathtt{Solution}[0]]$; }
\ENDIF
\RETURN{$\tilde{l}=\mathtt{Solution}[0]$;}
\end{algorithmic}
\end{algorithm}
In essence, the algorithm is executed twice - once for each set of option qudits $\{\ket{\psi(\theta_v)}\}_{\varepsilon N/2}$, where $v\in\{y,n\}$. It measures the states in each set and attributes to each one an integer. After all states have been measured, the algorithm creates a vector $\mathtt{Record}$, which contains the number of times each integer appeared during the measurements. Finally, Algorithm \ref{adval} creates the vector $\mathtt{Solution}$ in which it registers the values that appeared at least $40\%$ of times during the measurements, equivalently the values for which the $\mathtt{Record}$ vector assigned a number greater or equal than $40\%$ of times. The algorithm outputs the first value in the $\mathtt{Solution}$ vector. As we see in Figure \ref{example_Alorithm}, with high probability the value that algorithm outputs is either $l_{v}$ or $l_{v}-1$, for both values of $v$. Hence, we can find the difference $l_{y}-l_{n}$. After having acquired knowledge of $l_{y}-l_{n}$, the adversary can instruct the last corrupted voter to change the outcome of the voting process as previously described. 
\begin{figure}
\centering
\includegraphics[scale=0.95]{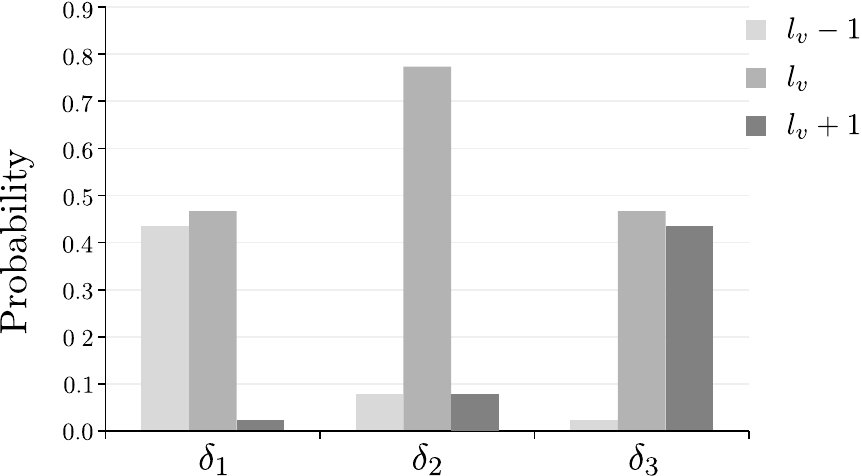}
\caption{The probabilities with which Algorithm \ref{adval} records a value in $\{l_v-1,l_v,l_v+1\}$ after measuring state $\ket{\psi(\theta_{v})}$ for $\delta_1=\frac{\pi}{2^{35}}$, $\delta_2=\frac{\pi}{2^{30}}$, and $\delta_3=\frac{\pi(2^6-1)}{2^{35}}$.}
\label{example_Alorithm}
\end{figure}

\noindent {\bf Probabilistic analysis} We prove here that the adversary's algorithm succeeds with overwhelming probability in $N$, where $N$ is the number of voters. Therefore, as we later prove in Theorem \ref{fexpN}, the election protocol needs to run at least exponentially many times with respect to $N$ in order to guarantee that the success probability of the adversary is at most $0.25$. We present here the necessary lemmas and give the full proofs in the Supplementary Material.
 
In order to compute the success probability of the attack, we first need to compute the probability of measuring a value in the interval $(x_l,x_{l+\w})$, where $x_{l}=\dfrac{2\pi l}{D}, ~~ l \in \{0,1,\ldots,D-1\}$\footnote{It is convenient to think of $l$ as the $D^{th}$ roots of unity.}. 

\begin{lemma}
\label{proof3}
Let $\Theta_{D,\delta}^{v} \in [0,2\pi]$ be the continuous random variable that describes the outcome of the measurement of an option qudit $\ket{\psi(\theta_{v})},v\in \{\text{y},\text{n}\}$ using operators: 
\begin{equation}\label{POVM}
E(\theta)=\frac{D}{2\pi} \ket{\Phi(\theta)} \bra{\Phi(\theta)}
\end{equation}
where $\ket{\Phi(\theta)}=\frac{1}{\sqrt{D}}  \sum\limits_{j=0}^{D-1} e^{ij\theta } \ket{j}$. 
 It holds that:
\begin{equation*}
\Pr[x_l<\Theta_{D,\delta}^{v}<x_{l+\w}]=\dfrac{1}{2\pi D} \int_{x_l}^{x_{l+\w}} \dfrac{\sin^2[D(\theta-\theta_{v})/2]}{\sin^2[(\theta-\theta_{v})/2]}d\theta
\end{equation*}
\end{lemma}

\noindent{}According to Algorithm \ref{adval}, an option qudit is attributed with the correct value $l_{v}$ when the result of the measurement is in the interval $[x_{l_{v}},x_{l_{v}+1}]$. Using Lemma \ref{proof3}, we can prove the following:

\begin{lemma}
\label{cint}
Let $\ket{\psi(\theta_{v})}$ be an option qudit of the protocol. Then it holds: $$\Pr[x_{l_v}<\Theta_{D,\delta}^{v}<x_{l_v+1}]> 0.405 $$ 

\end{lemma}

\noindent{}Lemma \ref{cint} shows that with probability at least $0.405$, the result of the measurement is in the interval $(x_{l_v},x_{l_v+1})$. Since Algorithm \ref{adval} inserts an integer to the $\mathtt{Solution}$ vector if it corresponds to at least $40\%$ of the total measured values, $l_v$ will most likely be included in the vector (we formally prove it later). Furthermore, we prove now that with high probability, there will be no other values to be inserted in $\mathtt{Solution}$, except the neighbours of the value $l_v$ (namely $l_v \pm 1 $).

\begin{lemma}
\label{proof0.9}
Let $\ket{\psi(\theta_{v})}$ be an option qudit of the protocol. Then it holds: $$Pr[x_{l_{v}-1}<\Theta_{D,\delta}^{v}<x_{l_{v}+2}]>0.9 $$

\end{lemma}

\noindent{}Here we need to note that we are aware of  the cases  $l_v\in\{0,D-1\}$ where the members $x_{l_{v}-1}$ and $x_{l_{v}+2}$ are not defined. It turns out not to be a problem and the same thing can be proven for these values (see Supplementary Material).

We have shown that the probability the measurement outcome lies in the interval $(x_{l_{v}-1},x_{l_{v}+2})$, and therefore gets attributed with a value of $l_v-1$, $l_v$ or $l_v+1$, is larger than $0.9$. If we treat each measurement performed by Algorithm \ref{adval} on each option qudit $\ket{\psi(\theta_v)}$, as an independent Bernoulli trial with success probability $p_{l}=Pr[x_{l}<\Theta_{D,\delta}^{v}<x_{l+1}] $, we can prove the following theorem:

\ignore{\begin{definition}
We can see each measurement that algorithm \ref{adval} performs at each vote state $\ket{\psi(\theta_v)}$, as an independent Bernoulli trial $X_{l}$ with probability of success $p_{l}=Pr[x_{l}<\Theta_{D,\delta}^{v}<x_{l+1}] $. Then the value of $\mathtt{Record}[l]$ follows  the binomial distribution:\[\mathtt{Record}[l]\sim B(\frac{\varepsilon N}{2},p_{l})\]

\end{definition}}

\begin{theorem}
\label{them1}
With overwhelming probability in the number of voters $N$, Algorithm \ref{adval} includes $l_v$ in the $\mathtt{Solution}$ vector
\[\Pr[\mathtt{Solution}[0]= l_v \vee \mathtt{Solution}[1]= l_v ]> 1-1/exp(\Omega (N))\]
\end{theorem}

\noindent We have proven that with overwhelming probability in $N$, integer $l_v$ occupies one of the two positions of vector  $\mathtt{Solution}$, but what about the other value? In the next theorem, we show that with overwhelming probability in $N$, the other value is one of the neighbours of $l_v$, namely $l_{v}+1$ or $l_{v}-1$.

\begin{theorem}
\label{Them2}
With negligible probability in the number of voters $N$, Algorithm \ref{adval} includes a value other than $(l_v-1,l_v, l_v+1)$ in the $\mathtt{Solution}$ vector, \emph{i.e.} $\forall \w \in \{0,\dots,l_{v}-2,l_{v}+2,\dots,D-1\} $: \[\Pr[\mathtt{Solution}[0]= \w \vee \mathtt{Solution}[1]=\w]<1/exp(\Omega (N)) \]
\end{theorem}

\begin{lemma}
\label{solutionst}
With overwhelming probability in $N$, the $\mathtt{Solution}$ vector in Algorithm \ref{adval}, is equal to $[l_{v}-1,l_{v}],[l_{v},"Null"]$ or $[l_{v},l_{v}+1]$. Specifically,
\begin{align*}
&\Pr[\mathtt{Solution} \in \{[l_{v}-1,l_{v}],[l_{v},"Null"],[l_{v},l_{v}+1]\}] > 1- 1/exp(\Omega (N))
\end{align*}
\end{lemma}

\noindent Now consider we have two executions of the Algorithm \ref{adval}, one for the "yes" and one for the "no" option qudits. It turns out that the values in the positions $l_{y}-1$ and $l_{n}-1$ of the vector $\mathtt{Record}$, follow the same Binomial distribution (it is easy to see that $p_{l_{y}-1}=p_{l_{n}-1}$). Also, each of them can be seen as a function of $\delta$  which is a monotonic decreasing function that takes a maximum value for  $\delta =0$ (the proof technique is similar to Lemma \ref{cint}). At this point the probability is equal to $p_{l_{v}}$, which is at least $0.405$ as we have proven in Lemma \ref{cint}\footnote{The same holds for the $p_{l_{y}+1},p_{l_{n}+1}$ except that probability is a monotonic increasing function with maximum value at point $\delta= 2\pi/D$ and value equal to $p_{l_{v}}$.}. Armed with this observation we can prove the next theorem.

\begin{theorem}
\label{OVERN}
If we define the event "Cheat" as:
\[
Cheat = \big [Algo(y)-Algo(n)=l_{y}-l_{n} \big ] 
\] where $Algo(v)$ is the execution of Algorithm \ref{adval} with $v\in\{y,n\}$, then it holds that: \[Pr["Cheat"]>1-1/exp(\Omega (N))\]
\end{theorem}

\begin{proof}
(sketch) We have seen that there exists a $\delta_0$ such that the probability $p_{l_{v}-1}$ is equal to $0.4$ for both values of $v$. It holds that: 
\begin{align*}
&Pr["Cheat"]=Pr["Cheat"|\delta \in [0,\delta_0)]\cdot Pr[\delta \in [0,\delta_0)] \nonumber \\ &+Pr["Cheat"|\delta =\delta_0]\cdot Pr[\delta=\delta_0]\nonumber \\
&+Pr["Cheat"|\delta \in (\delta_0,2\pi/D)]\cdot Pr[\delta \in (\delta_0,2\pi/D)]
\end{align*}
For the first interval, for both values of $v$, Algorithm \ref{adval} registers $\mathtt{Solution}=[l_{v}-1,l_{v}]$ with overwhelming probability in $N$. This holds because of Theorem \ref{solutionst} and the previous observation. Therefore, for both values of $v$ the algorithm outputs the values $l_{v}-1$. As a result, $l_{y}-1 -(l_n-1)=l_{y}-l_{n}$. 

For the second term, $ Pr[\delta=\delta_0]=0$, because $\delta$ is a continuous random variable. Finally, in the last term, the probability that the algorithm registers $\mathtt{Solution}=[l_{v}-1,l_{v}]$ is negligible in $N$, and by Theorem \ref{solutionst}, $\mathtt{Solution}$  has the form $[l_{v}]$ or $[l_{v},l_{v}+1]$. So for both values of $v$, the printed values are $l_{y}$ and $l_{n}$.

\end{proof}

At this point we have proven that the adversary succeeds with overwhelming probability in $N$ to perform the $d$-transfer attack in one round. But how many rounds should the protocol run in order to prevent this attack?

In the next theorem we prove that if the number of rounds is at most $exp(\Omega(N))$, the adversary succeeds with probability at least $0.25$. Although in a small election these numbers might not be big, in a large scale election it is infeasible to run the protocol as many times, making it either inefficient or insecure. We also note that the probabilistic analysis for one round is independent of the value $D$, so cannot be used to improve the security of the protocol.

\begin{theorem}
\label{fexpN}
Let $(\ket{\Phi}, \ket{\psi(\theta_y)},\ket{\psi(\theta_n)}, \delta, D, N)$ define one round of the protocol. If the protocol runs  $\rho$ rounds, where $2 \leq \rho \leq exp(\Omega(N)) $ , the $d$-transfer attack succeeds with probability at least $0.25$.
\end{theorem}

\begin{proof}
According to Theorem \ref{OVERN}  the probability that an adversary successfully performs the $d$-transfer attack is:
 \[Pr["Cheat"]>1-1/exp(\Omega(N))\]
Now, for $\rho$ protocol runs, where $2 \leq \rho \leq exp(\Omega(N)) $,  this probability becomes:
\begin{align*}
&(Pr["Cheat"])^{\rho}>(1-1/exp(\Omega(N)))^{\rho} \geq(1-1/\rho)^{\rho} > 0.25
\end{align*}

\end{proof}


Now, based on Theorem \ref{fexpN}, we can create an adversary such that $\mc{A}$ wins the $\mathbf{EXP}_{\msf{Qint}}^{\Pi}$ with probability at least $25\%$ if the protocol runs fewer than exponential number of rounds with respect to the number of voters.

\begin{theorem}
The adversary from section \ref{attackkvote} wins the experiment $\mathbf{EXP}_{\msf{Qint}}^{\Pi}$, where $\Pi$ is the protocol as described in section \ref{DistrDisc}, \ignore{$\ket{\msf{params}}=\ket{\Phi}\otimes \prod\limits_{j=1}^{N}(\ket{\psi(\theta_y)}_{j} \otimes \ket{\psi(\theta_n)}_{j}) \otimes \ket{\delta} \otimes \ket{D}
\otimes \ket{\rho}$,} with probability at least $0.25\%$ for every $\epsilon>0$ and number of rounds $2 \leq \rho \leq exp(\Omega(N))$.
\end{theorem}

\begin{proof}
(sketch) When $\mc{A}$ sends the ballot to $\mc{C}$ on behalf of a corrupted voter in \textbf{Casting phase}, $\mc{A}$ applies the operations as described in section \ref{attackkvote}. Next, in \textbf{Tally phase} $\mc{C}$ computes the election outcome. If in a round the result is different from a previous round, $\mc{C}$ outputs "$\bot$". However, from Theorem \ref{fexpN}, we know that $\mc{C}$ will not output "$\bot$"  with probability at least $0.25\%$.
\end{proof}


\section{Quantum voting based on conjugate coding}
\label{okamoto}
This section looks at protocols based on conjugate coding (\cite{OY08,ZR13}). The participants in this family of protocols are one or more election authorities\ignore{\footnote{In \cite{ZR13} the authors introduced two election authorities in order to distribute the trust between them.}}, the tallier and the voters. The election authorities are only trusted for the purpose of eligibility; privacy should be guaranteed by the protocol against both malicious $EA$ and $\A$. Unlike the previous protocols, here the voters do not share any entangled states with neither $EA$ nor $\A$ in order to cast their ballots.
One of the main differences between the two protocols is that \cite{OY08} does not provide any verification of the election outcome, while \cite{ZR13} does, but at the expense of receipt freeness, which \cite{OY08} satisfies. Specifically, in \cite{ZR13} each $V_{k}$ establishes two keys with $T$ in an anonymous way by using part of protocol \cite{OY08} as a subroutine. It's worth to mention that in order for these keys to be established, further interaction between the voters and $EA$ is required and $EA$ is assumed trusted for that task. At the end of an execution, $V_{k}$ encrypts the ballot with one of the keys and sends it to $T$ over a quantum anonymous channel. $T$ announces the result of each ballot accompanied with the second key so that the voters can verify that their ballot has been counted. This makes it also possible for a coercer to verify how a voter voted, by showing them the second key used as a receipt.
It is worth mentioning that protocol \cite{OY08} could easily be made to satisfy the same notion of verifiability.

\subsection{Protocol Specification}

\noindent{\bf Set up phase}
\begin{enumerate}
\item $EA$ picks a vector $\bar{b}=(b_{1},\dots,b_{n+1}) \in_{R} \{0,1\}^{n+1}$, where $n$ is the security parameter of the protocol. This vector will be used by $EA$ for the encoding of the ballots  and it will be kept secret from $T$ until the end of the ballot casting phase.
\item For each $V_k$, $EA$ prepares $w=poly(n)$ blank ballot fragments each of the form $\ket{\phi_{\bar{a}_{j},\bar{b}}}=\ket{\psi_{a_{j}^{1},b_{1}}} \otimes \ldots \otimes \ket{\psi_{a_{j}^{n+1},b_{n+1}}}, j \in \{1,\ldots,w\}$, where $\bar{a}_{j}=(a_{j}^{1},\ldots,a_{j}^{n+1})$ such that: 
\begin{align*}
(&a_{j}^{1},\ldots,a_{j}^{n}) \in_{R} \{0,1\}^{n}, a_{j}^{n+1}= a_{j}^{1}\oplus \ldots \oplus a_{j}^{n}
\end{align*}

and: ~~~~~$\ket{\psi_{0,0}}=\ket{0}, \ket{\psi_{1,0}}=\ket{1},\ket{\psi_{0,1}}= \frac{1}{\sqrt{2}}(\ket{0}+\ket{1}), \ket{\psi_{1,1}}= \frac{1}{\sqrt{2}}(\ket{0}-\ket{1})$.\\
These $w$ fragments will constitute a blank ballot (e.g the first row of Fig. \ref{example_Conjugate_Coding} is a blank ballot fragment).

\item $EA$ sends one blank ballot to each $V_k$ over an authenticated channel.
\end{enumerate}

\noindent{\bf Casting phase}
\begin{enumerate}\setcounter{enumi}{3}
\item After reception of the blank ballot, each $V_k$ re-randomizes it by picking for each fragment a vector $\bar{d}_{j}=(d_{j}^{1},\ldots,d_{j}^{n+1})$ such that:
\begin{align*}
(&d_{j}^{1},\ldots,d_{j}^{n}) \in_{R} \{0,1\}^{n}, d_{j}^{n+1}= d_{j}^{1}\oplus \ldots \oplus d_{j}^{n}.
\end{align*}
$\forall j \in \{1,\ldots, w\}$, $V_k$ applies unitary $U_{j}^{\bar{d}_{j}}=Y^{d_{j}^{1}} \otimes \ldots \otimes Y^{d_{j}^{n+1}}$ to the blank ballot fragment $\ket{\phi_{\bar{a}_j,\bar{b}}},$, where:
\[Y^{1}=\begin{bmatrix}
0 & -1 \\ 1 & 0
\end{bmatrix} , Y^{0}= \mathbb{I}\]
\item $V_{k}$ encodes the candidate of choice in the ($n+1$)$^{th}$-qubit of the last blank ballot fragments\footnote{Candidate choices are encoded in binary format.}. For example, if we assume a referendum type election, $V_{k}$  votes for $c \in \{0,1\}$ by applying to the blank ballot fragment $\ket{\phi_{\bar{a}_w,\bar{b}}}$ the unitary operations $U_{w}^{\bar{c}} $ respectively, where: $\bar{c}=(0,\ldots,0,c)$ (see Fig. \ref{example_Conjugate_Coding}).
\item $V_{k}$ sends the ballot to $T$ over an anonymous channel.
\end{enumerate}

\noindent{\bf Tally phase}
\begin{enumerate}\setcounter{enumi}{6}
\item Once the ballot casting phase ends, $EA$ announces $\bar{b}$ to $T$.
\item  With this knowledge, $T$ can decode each cast ballot in the correct basis. Specifically, $T$ decodes each ballot fragment by measuring it in the basis described by vector $\bar{b}$ and XORs the resulting bits. After doing this to each ballot fragment, $T$ ends up with a string, which is the actual vote cast.

\item $T$ announces the election result. 
\end{enumerate}
\begin{figure}
\centering
\includegraphics[scale=0.95]{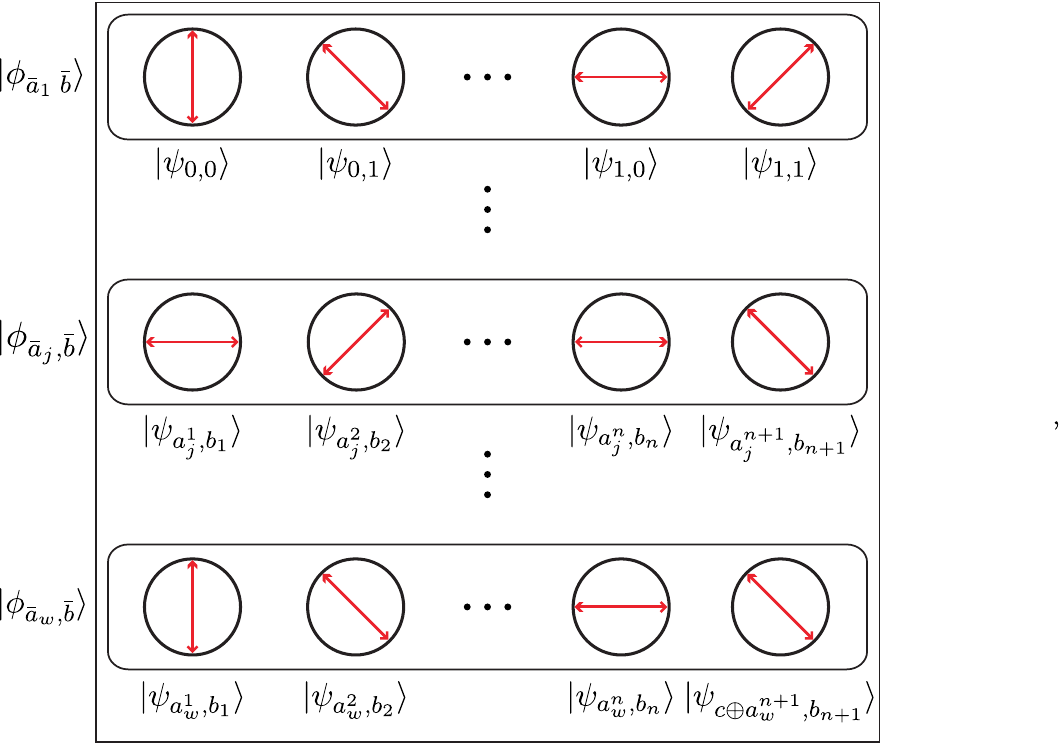}
\caption{The ballot consisting of $w$ ballot fragments, which encode the binary choice ``0\ldots 01" in a referendum type election example.}
\label{example_Conjugate_Coding}
\end{figure}

\subsection{Vulnerabilities of Conjugate Coding Protocols}
The technique underlying this protocol is closely related to the one used in the first quantum key distribution protocols (\cite{BB84,SJ00}). However, it has some limitations in the context of these voting schemes.
\medskip

\noindent {\bf Malleable blank ballots:} An adversary can change the vote of an eligible voter, when the corresponding ballot is cast over the anonymous channel. Assume $V_{k}$ has applied the appropriate unitary on the blank ballot in order to vote for the candidate of their choice. And let us consider that the last $m$ ballot fragments encode the candidate. When the adversary sees the cast ballot over the quantum anonymous channel, they apply the unitary $U_{w-(m-1)}^{\bar{c}_{1}}, \ldots, U_{w}^{\bar{c}_{m}}$, where $c_{r}$ is either $0$ or $1$, depending on their choice to flip the candidate bit or not. As a result the adversary modifies the ballot of $V_{k}$ such that it decodes to a different candidate than the intended one. This is possible because the adversary is aware of the ballot fragments used to encode the candidate choice. Furthermore, if the adversary has side channel information about the likely winning candidate (from pre-election polls for instance), they will be able to change the vote encoded in the ballot into one of their desire. This is possible because the adversary is aware of which bits are encoded in the ballot more frequently and knows exactly which unitary operator to apply in order to decode to a specific candidate.

\noindent {\bf Violation of privacy:} It is already acknowledged by the authors of \cite{OY08} that the $EA$ can introduce a ``serial number" in a blank ballot to identify a voter, \emph{i.e} some of the blank ballot fragments in the head of the ballot  decode to ``1" instead of ``0". This allows the $EA$ to decode any ballot cast over the quantum anonymous channel, linking the identity of the voters with their choice. 

\ignore{\noindent {\bf Double vote by keeping some  ballot fragments:} An adversary controlling $t$-out-of-$N$ voters can double vote by keeping some of their ballot fragments for future use. Specifically, suppose that we need $m$ ballot fragments in order to encode a candidate choice.  Each of the corrupted voters proceeds as follow:
\begin{enumerate}
\item Upon reception of the $w$ blank ballot fragments from the $EA$, corrupted $V_{k}$ keeps the last $m$ blank ballot fragments. 
\item Then $V_k$ sends to $T$ over the quantum anonymous channel the remaining $w-m$ blank ballot fragments accompanied with $m$ fake ballot fragments encoded in a random basis. When  $T$ decodes this ballot the last $m$ states open to a random bit sequence, but the remaining $w-m$ open to the value $0$. As a result it is accepted by $T$. Therefore, $V_k$'s vote, even though favors a random candidate, is accepted by $T$. By the end of this procedure, the adversary has $tm$ blank ballots fragments. 
\item If $tm>w$, the adversary has at its disposal one extra blank ballot.
\end{enumerate}
A possible fix for this can be to increase the parameter of the protocol so than $Nm<<w$.\NL{if I can vote for 2 random candidates and I don't have the control over the votes and both votes are counted, is double voting}}

\noindent {\bf One-more-unforgeability:} The security of the protocol relies on a quantum problem introduced in \cite{OY08}, named \emph{one-more-unforgeability} and the assumption that it is computationally hard for a quantum adversary. The game that captures this assumption goes as follows: a challenger encodes $w$ blank ballot fragments in a basis $\bar{b}$ and gives them to the adversary. The adversary wins the game if they produce $w+1$ valid blank ballot fragments in the basis $\bar{b}$. The authors claim the probability of the adversary of winning this game is at most $1/2 + 1/2(negl(n))$. 

\noindent {\bf On the security parameter:} Because of the ballots' malleability, an adversary could substitute the parts of the corrupted voters' blank ballot fragments that encode a candidate, with blank ballot fragments in a random base. Of course these ballots would open into random candidates in a specific domain but would still be valid, since the leading zeros would not be affected by this change. This is because blank ballots contain no entanglement. Now the adversary can keep these valid spare blank ballot fragments to create new valid blank ballots. To address this problem, the size of blank ballots needs to be substantially big compared to the number of voters and the size of the candidate space ($Nm<<w$).

\section{Other Protocols}
\label{ctrl}

Other protocols have also been proposed, with main characteristic that the $EA$ controls when ballots get counted. This can be achieved with either the use of shared entangled states between $EA$, $T$, and $V_{k}$ \cite{TK17,XP17} or Bell pairs \cite{TK17} between $T$ and $V_{k}$ with $EA$ knowing the identity of the holder of each pair particle. However, we do not fully analyse these protocols in this review, as they have many and serious flaws making even the correctness arguable. The protocol of \cite{XP17} claims to provide verifiability of the election outcome, but without explaining how this can be achieved. From our understanding of the protocol this seems unlikely to be the case. From the description of the protocol each voter can change their mind and announce a different vote from the originally cast one. This is possible because the function every voter uses to encode their vote is not committed in any way. Two protocols introduced in \cite{TK17} have similar limitations. For instance, there is no mechanism for verifiability of the election outcome. In addition, privacy against $T$ is not satisfied in contrast to protocols we saw in section \ref{distributed}. This is because each voter's vote is handled individually and not in a homomorphic manner. All of these could be achieved just by a classical secure channel. Last, the protocol appearing in \cite{HD11} shares many of the limitations of the former protocols as well as some further ones. The method introduced for detecting eavesdropping in the election process is insecure, as trust is put into another voter in order to detect any deviation from the protocol specification. Moreover, the way each voter casts their vote is not well defined in the protocol, which makes privacy and correctness trivially violated.


\section{Discussion}
\label{discussion}

In this work, we have examined the current state of the art in quantum e-voting, by presenting the most prominent proposals and analyzing their security. What we have found is that all the proposed protocols fail to satisfy the necessary security definitions for future implementations. Despite this, these protocols open the way to new avenues of research, specifically on whether quantum information can solve some long-standing issues in e-voting and cryptography in general. By studying them, we can identify several interesting ideas for further development as well as possible bottlenecks in future quantum protocols. 

For instance, we saw that, unless combined with some new technique, the traveling ballot protocols do not provide a viable solution, as double-voting is always possible, and there is no straightforward way to guarantee privacy. On the other hand, the distributed ballot protocols give us very strong privacy guarantees because of the entanglement between the ballot states, but it seems that verifiability against malicious talliers might be hard to achieve. In fact, one of the most intriguing questions in quantum e-voting is whether we can achieve all desired properties simultaneously. For instance, every classical definition of verifiability~\cite{CV16} assumes a trusted bulletin board that the participants can read, write on, and finally verify the outcome of the election. However, implementing a quantum bulletin board to achieve the same properties is not straightforward, since  reading a quantum state can 'disturb' it in an irreversible way. For that reason we have defined in our experiment $\mathbf{EXP}_{\msf{QVer}}^{\Pi}$ the public quantum register $\mc{B}$, and the predicate $\mathsf{Verify}^{\Pi}$. Of course, an implementation of such predicate seems hard to realize in the quantum setting and more research is needed.

We have also shown that the cut-and-choose technique used by the protocols in Section \ref{DualBase} is both inefficient and insecure. A solution could be to provide some type of randomness to the voters (e.g. in the form of a common random string), which would determine if a state should be verified or used for the voting phase (a similar process is shown in \cite{McCutcheon16,Pappa12}). However, even if the problem with the cut-and-choose technique is addressed, privacy can still be violated as we have seen, and possible corrections might require the use of more advanced techniques. Notwithstanding these limitations, we believe that our analysis opens new research directions for the study of the quantum cut-and-choose technique, which plays a fundamental role in the secure distribution of quantum information.

The general aim of studying quantum cryptographic protocols, is to provide better guarantees than classically possible, be that in security or efficiency. This has been achieved for primitives like coin flipping and oblivious transfer, against unbounded adversaries \cite{CK10,Pappa14}, and against bounded ones that are more relevant to practical implementations (e.g. with limited storage \cite{DFSS05}, noisy storage \cite{KWW12}, or bounded by relativistic constraints \cite{LKB15}). The question whether quantum technology could enhance electronic voting as well has not yet been answered, and requires further study of both the existing classical and quantum literature. First, bottlenecks in classical election protocols that could potentially be solved by using quantum subroutines, need to be identified. Then, quantum protocols need to be designed, that satisfy well articulated definitions of all the required properties in composable frameworks. In classical cryptography, this was pursued with the help of automated provers and model checkers such as EasyCrypt \cite{Cortier2017}, game based definitions \cite{CV16,BD15}, and by employing the Universal Composability Framework \cite{CR01,GJ04}. However, in quantum cryptography, it remains unclear how these techniques can be adopted. An interesting approach appears in \cite{QMC}, where the authors provide an automated verification tool that enables checking properties of systems which can be expressed within the quantum stabilizer formalism. Finally, a recent work by Unruh \cite{unruh2018} on quantum relational Hoare logic might open new avenues and help provide a solution to this problem.

\bibliographystyle{splncs04}

\newpage


\appendix

\section*{Supplementary Material}\label{sec:supplementary}

In the supplementary material we include some extra technical details that due to lack of space could not be included in the main body of the paper, as well as our response to reviews received on previous submissions of this work.

\section{Formal definition of quantum integrity}
\label{app:integrity-definition}

The integrity experiment $\mathbf{EXP}_{\msf{Qint}}^{\Pi}$ is the same as $\mathbf{EXP}_{\msf{Qver}}^{\Pi}$ with the only exceptions that there isn't a predicate $\msf{P}_{\msf{Verify}}^{\Pi}$ and we don't allow $\mc{A}$ to corrupt the tallier (if there exist in the protocol $\Pi$). As a result, we don't capture universal verifiability in $\mathbf{EXP}_{\msf{Qint}}^{\Pi}$, but only double voting and vote deletion/alteration of honest ballots.

\begin{boxfigH}{The experiment $\mathbf{EXP}_{\msf{Qint}}^{\Pi}$
}{fig:Qint}
\underline{\emph{The experiment $\mathbf{EXP}_{\msf{Qint}}^{\Pi}(\mathcal{A},\epsilon,\delta_{0})$}}
\begin{itemize}
\item[--]\textbf{Set up phase:} $\mc{C}$ and $\mc{A}$ generate the protocol parameters in quantum register $\mc{X}$ as specified by $\Pi$. Furthermore, $\mc{A}$ chooses the votes for all voters $\{v_{k}\}_{V_{k} \in \mathcal{V}}$. 
\item[--] \textbf{Casting phase:} $\mc{A}$ chooses whether to corrupt $V_{\pe(k)}$ or not (and therefore add the voter or not to the set $\mc{V}_{\mc{A}})$.
\begin{itemize}

\item If $V_{\pe(k)}\not\in\mc{V}_{\mc{A}}$, $\mc{C}$ generates the ballot $\{\mc{B}_{\pe(k)},\bot\} \leftarrow \mathsf{CastBallot}(v_{\pe(k)},\mc{X}_{\pe(k)},\mc{B},\delta_{0})$. If it is not $\bot$, $\mc{C}$ sends it to $\mc{A}$. If $\mc{C}$ receives $\mc{B}_{\pe(k)}$ back from $\mc{A}$, then $\mc{C}$ stores $\mc{B}_{\pe(k)}$ in $\mc{B}$, where $\{\mc{B}_{\pe(k)}$, $\mc{X}_{\pe(k)}\}$, and $\mc{B}$ are local and global quantum registers respectively. Note that when $\mc{A}$ receives $\mc{B}_{\pe(k)}$ from $\mc{C}$, it is possible to apply quantum operations on the register that are dependent on the specifications of $\Pi$. 
\item If $V_{\pe(k)}\in\mc{V}_{\mc{A}}$, then $\mc{A}$ creates a ballot $\mc{B}_{\pe(k)}$ and sends it to $\mc{C}$.
\end{itemize}

\item[--] \textbf{Tally phase:} $\mc{C}$ computes $X \leftarrow \mathsf{Tally}(\mc{B},\mc{X}_{\mc{C}},\delta_{0}) $:
\begin{itemize}
\item If $X \not= \bot $ and $\msf{P}_{\msf{VCounted}}^{\Pi}(\{v_{k}\}_{V_{k} \not\in \mathcal{V}_{\mc{A}}},X) = 0$ or $\mathsf{Nballots}^{\Pi}(X)> |\mc{V}|$, then {\tt output 1}, else {\tt output 0}.
\end{itemize}
\end{itemize}
\end{boxfigH}

\begin{definition}
We say that a quantum e-voting protocol $\Pi$ satisfies $\epsilon$-{\bf quantum integrity} if for every quantum PPT $\mc{A}$ the probability to win the experiment $\mathbf{EXP}_{\msf{Qint}}^{\Pi}(\mathcal{A},\epsilon,\delta_{0})$ is negligible with respect to $\delta_{0}$:

\[\Pr[1\leftarrow \mathbf{EXP}_{\msf{Qint}}^{\Pi}(\mathcal{A},\epsilon,\delta_{0})]=negl(\delta_{0}).\]
\end{definition}

\section{Proof of attack on Distributed Ballot protocols}
\label{app:proofs}
Now we give detailed proofs of the theorems and lemmas of Section~\ref{distributed}.

\begin{customlemma}{\ref{proof3}}
Let $\Theta_{D,\delta}^{v} \in [0,2\pi]$ be the continuous random variable that describes the outcome of the measurement of a vote state $\ket{\psi(\theta_{v})},v\in \{\text{y},\text{n}\}$ using operators 
\begin{equation}
E(\theta)=\frac{D}{2\pi} \ket{\Phi(\theta)} \bra{\Phi(\theta)}
\end{equation}
where $\ket{\Phi(\theta)}=\frac{1}{\sqrt{D}}  \sum\limits_{j=0}^{D-1} e^{ij\theta } \ket{j}$. It holds that:
\begin{equation}\label{eqn:prob}
Pr[x_l<\Theta_{D,\delta}^{v}<x_{l+\w}]=\dfrac{1}{2\pi D} \int_{x_l}^{x_{l+\w}} \dfrac{\sin^2[D(\theta-\theta_{v})/2]}{\sin^2[(\theta-\theta_{v})/2]}d\theta
\end{equation}
\end{customlemma}

\begin{proof}
\begin{align*}
&Pr[x_l<\Theta_{D,\delta}^{v}<x_{l+\w}]=\bra{\phi(\theta_{v})}\int_{x_l}^{x_{l+\w}}E(\theta)d\theta\ket{\phi(\theta_{v})}\\&=\int_{x_l}^{x_{l+\w}}\bra{\phi(\theta_{v})}E(\theta)\ket{\phi(\theta_{v})}d\theta \\ &=\frac{D}{2\pi D^2} \int_{x_l}^{x_{l+\w}} |\sum\limits_{j=0}^{D-1} e^{(\theta-\theta_{v})ij} |^2  d\theta \\ &=\frac{1}{2\pi D}\int_{x_l}^{x_{l+\w}}([\sum\limits_{j=0}^{D-1}\cos[(\theta -\theta_{v})j]]^2 \\ &+[\sum\limits_{j=0}^{D-1}\sin[(\theta -\theta_{v})j]]^2)d\theta 
\end{align*}
For any $x\in\mathbb{R}$, the following two equations hold:
$$\sum\limits_{j=0}^{D-1}\cos[jx]=\dfrac{\sin[Dx/2]}{\sin[x/2]} \cos[(D-1)x/2]$$

$$\sum\limits_{j=0}^{D-1}\sin[jx]=\dfrac{\sin[Dx/2]}{\sin[x/2]} \sin[(D-1)x/2]$$
So finally we have:
\begin{align*}
Pr[x_l<\Theta_{D,\delta}^{v}<x_{l+\w}]= \frac{1}{2\pi D} \int_{x_l}^{x_{l+\w}} \frac{\sin^2[D(\theta-\theta_{v})/2]}{\sin^2[(\theta-\theta_{v})/2]}d\theta
\end{align*}

\end{proof}

\begin{customlemma}{\ref{cint}}
Let $\ket{\psi(\theta_{v})}$ be a voting state of the protocol. Then it holds: $$Pr[x_{l_v}<\Theta_{D,\delta}^{v}<x_{l_v+1}]> 0.405 $$ 

\end{customlemma}

\begin{proof}
A simple change of variables in Eq.\eqref{eqn:prob} gives us:

\begin{equation*}
Pr[x_{l_v}<\Theta_{D,\delta}^{v}<x_{l_v+1}]= \frac{1}{2\pi D}\int_{0}^{2\pi/D} \frac{\sin^2[D(\theta-\delta)/2]}{\sin^2[(\theta-\delta)/2]}d\theta
\end{equation*}

\noindent By setting $(\theta-\delta)/2=y$, we get: 
\begin{equation*}
Pr[x_{l_{v}}<\Theta_{D,\delta}^{v}<x_{l_{v}+1}]= \frac{1}{\pi D}\int_{-\delta/2}^{(2\pi/D -\delta)/2}\frac{\sin^2[Dy]}{\sin^2[y]} dy
\end{equation*}

The above is just a function of $\delta$, which we denote as $F(\delta)$. In order to lower-bound $F(\delta )$ we need to find its derivative: 
\begin{equation*}
\dfrac{dF(\delta )}{d\delta}=\frac{1}{2\pi D} \Bigg(\frac{\sin^2[D\delta/2]}{\sin^2[\delta/2]}-\frac{\sin^2[D\delta/2]}{\sin^2[(2\pi/D-\delta)/2]}\Bigg)
\end{equation*}
It is easy to check that: 
\begin{align*}
\dfrac{dF(\delta )}{d\delta }=0&,~~ \text{when}~~ \delta=0 ~~\text{or}~~ \delta= \pi /D  \\
\dfrac{dF(\delta )}{d\delta }>0&,~~ \text{when} ~~ 0<\delta<\pi/D \\ 
\dfrac{dF(\delta )}{d\delta }<0&,~~ \text{when} ~~ \pi/D<\delta<2\pi/D
\end{align*}
It also holds that $F(0)=F(2\pi/D)$, so the minimum extreme points  of our function are equal. As a result we have:
\begin{align}
F(\delta )\geq \displaystyle \lim_{\delta \to 0^{-}}F(\delta )=F(0)
\end{align}
From the fact that: 
\begin{align*}
&|\sin[x]| \leq |x|, \forall x \in \mathbb{R} \\ & |\sin[x]|\geq |(2/\pi) x|,\forall x\in [0,\pi/2]  \\ &|\sin[x]|\geq |-( 2/\pi) x+2|,\forall x\in [\pi/2,\pi]
\end{align*}
It follows:
\begin{align*}
F(0)  &\geq \frac{1}{\pi D} \int_{0}^{\frac{\pi}{2D}} \Big(\frac{2}{\pi Dy}\Big)^2/y^2 dy +\int_{\frac{\pi}{2D}}^{\frac{\pi}{D}} \Big(\frac{2}{\pi Dy}+2\Big)^2/y^2 dy \\ &\geq \frac{4}{\pi^{2}} \\ & > 0.405
\end{align*}
\end{proof}
Now in order to prove lemma \ref{proof0.9}, we need the following proposition:

\begin{proposition}\label{boundsin}
$\forall x\in [-2\pi,2\pi]$ it holds that: 
\begin{equation}\label{eqn:sin}
\sin^2[x]>\sum\limits_{n=1}^{20} (-1)^{n+1} \dfrac{2^{2n-1}x^{2n}}{(2n)!}
\end{equation}
\end{proposition}

\begin{proof}
From the Taylor series expansion at point $0$ of $\cos[x]$, we know that:
\[
 \cos[x]=\sum\limits_{n=0}^{\infty}(-1)^n\dfrac{x^{2n}}{(2n)!}, ~~\forall x\in \mathbb{R} 
\]
Then:
\begin{align*}
\sin^2[x] &=\frac{1}{2}-\frac{\cos[2x]}{2} =\frac{1}{2} -\frac{1}{2}\sum\limits_{n=0}^{\infty}(-1)^n \frac{2^{2n}x^{2n}}{(2n)!}  \nonumber \\
& = \sum\limits_{n=1}^{\infty} (-1)^{n+1} \frac{2^{2n-1}x^{2n}}{(2n)!}
\end{align*}
Given the above equation, in order to prove Eq.\eqref{eqn:sin}, we simply need to show:
\begin{align*}
\sum\limits_{n=21}^{\infty} (-1)^{n+1} \dfrac{2^{2n-1}x^{2n}}{(2n)!}>0
\end{align*}
If we think of the above as a sum of terms $a_n$ ($n=21,\dots,\infty)$, for integer $j\geq 10$, it holds that:
\begin{align*}
a_n>0, ~&\text{when}~n=2j+1, \\ a_n<0, ~&\text{when}~n=2j.
\end{align*}
We therefore need to prove that $\sum\limits_{n=21}^{\infty} a_n>0$, which in turn is equivalent to proving that:
\begin{align*}
|a_{n}|>|a_{n+1}|  &\Longleftrightarrow 2^{2n-1}x^{2n}/(2n)! > 2^{2n+1}x^{2n+2}/(2n+2)! \\ & \Longleftrightarrow 1 > 4x^{2}/((2n+1)(2n+2)) \\ & \Longleftrightarrow (2n+1)(2n+2)/4>x^2
\end{align*}
In this case, the above holds, because the minimum value of n is $21$ and the maximum value of $x^2$ is $4\pi^2$.

\end{proof}

\begin{customlemma}{\ref{proof0.9}}
Let $\ket{\psi(\theta_{v})}$ be a voting state of the protocol. Then it holds: $$Pr[x_{l_{v}-1}<\Theta_{D,\delta}^{v}<x_{l_{v}+2}]>0.9 $$

\end{customlemma}

\begin{proof}
We follow exactly the same procedure as lemma \ref{cint} and get:
\begin{align}
&Pr[x_{l_{v}-1}<\Theta_{D,\delta}^{v}<x_{l_{v}+2}]\\ &=\frac{1}{2\pi D} \int_{x_{l_{v}-1}}^{x_{l_{v}+2}} \frac{\sin^2[D(\theta-\theta_{v})/2]}{\sin^2[(\theta-\theta_{v})/2]}d\theta\nonumber \\ 
&= \frac{1}{2\pi D}\int_{-2\pi/D}^{4\pi/D} \frac{\sin^2[D(\theta-\delta)/2]}{\sin^2[(\theta-\delta)/2]}d\theta \nonumber  \\
&= \frac{1}{\pi D}\int_{-\pi/D-\delta/2}^{2\pi/D -\delta/2}\frac{\sin^2[Dy]}{\sin^2[y]} dy
\end{align}
where $(\theta-\delta)/2=y$. Again the above probability depends only on $\delta$ and can therefore be denoted with $F(\delta)$. In a similar way as before, we can prove that the minimum of this function is at $\delta=0$ and compute $F(0)$.

\begin{align}
F(0) &= \frac{1}{\pi D}\int_{-\pi/D}^{2\pi/D }\frac{\sin^2[Dy]}{\sin^2[y]} dy \nonumber \\ 
     & \geq \frac{1}{\pi D} \int_{-\pi/D}^{2\pi/D } \frac{\sum\limits_{n=1}^{20}\frac{(-1)^{n+1}2^{2n-1} (Dy)^{2n}}{(2n)!}}{y^2} dy \nonumber \\ 
     &= \frac{1}{\pi D} \sum\limits_{n=1}^{20}\int_{-\pi/D}^{2\pi/D}\frac{(-1)^{n+1}2^{2n-1}D^{2n}y^{2n}}{y^2 (2n)!} dy \nonumber \\ 
     &= \frac{1}{\pi D} \sum\limits_{n=1}^{20} \frac{(-1)^{n+1} 2^{2n-1}D^{2n}}{(2n)!} \int_{-\pi /D}^{2 \pi /D}y^{2(n-1)} dy \nonumber \\ 
     &= \frac{1}{\pi D} \sum\limits_{n=1}^{20} \frac{(-1)^{n+1} 2^{2n-1}D^{2n}}{(2n)!} [y^{2n-1}/(2n-1)]_{-\pi/D}^{2\pi/D} \nonumber \\ 
     &= \sum\limits_{n=1}^{20}\frac{(-1)^{n+1}2^{2n-1}}{(2n)!} \frac{\pi ^{2n-2} (2^{2n-2}+1)}{2n-1} \nonumber \\ &\approx 0.9263
\end{align}

\end{proof}

\begin{customtheorem}{\ref{them1}}
With overwhelming probability in the number of voters $N$, algorithm \ref{adval} includes $l_v$ in the $\mathtt{Solution}$ vector (i.e. it measures a value in the interval $[x_{l_v},x_{l_v+1}]$ more than 40\% of the time). 
$$\Pr[\mathtt{Solution}[0]= l_v \vee \mathtt{Solution}[1]= l_v ]> 1-1/exp(\Omega (N))$$
\end{customtheorem}

\begin{proof}
We can see each measurement that algorithm \ref{adval} performs at each vote state $\ket{\psi(\theta_v)}$, as an independent Bernoulli trial $X_{l}$ with probability of success $p_{l}=\Pr[x_{l}<\Theta_{D,\delta}^{v}<x_{l+1}] $. Then the value of $\mathtt{Record}[l]$ follows  the binomial distribution:$$X_{\mathtt{Record}[l]}\sim B(\frac{\varepsilon N}{2},p_{l})$$

We can therefore compute:
\begin{align*}
&\Pr\big[\mathtt{Solution}[0]= l_v \vee \mathtt{Solution}[1]= l_v\big]\\\ 
&=\Pr\big[\mathtt{Record}[l_v] \geq 0.4 \varepsilon N/2\big] \\
&\geq 1- \Pr\big[\mathtt{Record}[l_v] \leq 0.4 \varepsilon N/2\big] \\ 
&\overset{\ref{test}}{=} 1- \Pr\big[\mathtt{Record}[l_v] \leq (1-\gamma) p_{l_{v}}\varepsilon N/2\big] \\ 
&\overset{\ref{test2}}{\geq}1- exp(-\gamma^2 p_{l_{v}} \varepsilon N/6) \\ 
&=1- (exp(-\gamma^2 p_{l_{v}} \varepsilon /6))^N \\ 
&=1-1/exp(\Omega (N))
\end{align*}
\end{proof}

\footnotetext[1]{\label{test}$p_{l_{v}}>0.405\Longrightarrow \exists \gamma>0$ s.t $0.4=(1-\gamma) p_{l_{v}}$}

\footnotetext[2]{\label{test2}The Chernoff  bound for a random variable $X\sim B(N,p)$ and expected value $E[X]=\mu$ is: $Pr[X \leq (1-\gamma)\mu]\leq exp(-\gamma^2\mu/3)$}

\begin{customtheorem}{\ref{Them2}}
With negligible probability in the number of voters $N$, algorithm \ref{adval} includes a value other than $(l_v-1,l_v, l_v+1)$ in the $\mathtt{Solution}$ vector, i.e. $\forall \w \in \{0,\dots,l_{v}-2,l_{v}+2,\dots,D-1\} $: $$Pr[\mathtt{Solution}[0]= \w \vee \mathtt{Solution}[1]=\w]<1/exp(\Omega (N)) $$
\end{customtheorem}

\begin{proof}
\label{proofnea}
Let $\w \in \{0,\ldots,D-1\}\setminus\{l_{v}-1,l_{v},l_{v}+1\}$, then it holds:
\begin{align*}
&Pr[\mathtt{Solution}[0]= \w \vee \mathtt{Solution}[1]=\w]\\ &=Pr[X_{\mathtt{Record}[\w]} \geq 0.4 \varepsilon N/2]
\end{align*}
We know from lemma \ref{proof0.9} that $p_{\w}<0.1$, so $\exists \gamma>0$ such that:\footnote{The Chernoff  bound for a random variable $X\sim B(N,p)$ and expected value $E[X]=\mu$ is: $Pr[X \leq (1+\gamma)\mu]\leq exp(-\gamma\mu/3), \gamma >1$}
\begin{align*}
&Pr[X_{\mathtt{Record}[\w]} \geq 0.4 \varepsilon N/2] \\ &= Pr[X_{\mathtt{Record}[\w]} \geq (1+\gamma)p_{\w} \varepsilon N/2] \\ & < exp(- \gamma p_{\w} \varepsilon N/6) \\ &= (exp(-\gamma p_{\w} \varepsilon/6))^N \\ &= 1/exp(\Omega (N))
\end{align*}
\end{proof}

\begin{customlemma}{\ref{solutionst}}
With overwhelming probability in $N$, the $\mathtt{Solution}$ vector in algorithm \ref{adval}, is equal to $[l_{v}-1,l_{v}],[l_{v},``Null"]$ or $[l_{v},l_{v}+1]$.Specifically,
\begin{align*}
&\Pr[\mathtt{Solution} \in \{[l_{v}-1,l_{v}],[l_{v},``Null"],[l_{v},l_{v}+1]\}] \\  &> 1- 1/exp(\Omega (N))
\end{align*}
\end{customlemma}

\begin{proof}
Let as define the following events: 
\begin{align*}
A=\big[&\mathtt{Solution}[0]= \w \vee \mathtt{Solution}[1]=\w,\\ & \w \in \{0,\dots,l_{v}-2,l_{v}+2,\dots,D-1\}\big]
\end{align*}
\begin{align*}
B=\big[\mathtt{Solution}[0]= l_{v} \vee \mathtt{Solution}[1]=l_{v}\big]
\end{align*}

Since the cases $\mathtt{Solution}=[l_{v},l_{v}-1]$ and $\mathtt{Solution}=[l_{v}+1,l_{v}]$ are impossible from the construction of the algorithm, from theorems \ref{them1} and  \ref{Them2} it holds:
\begin{align*}
&\Pr[\mathtt{Solution} \in \{[l_{v}-1,l_{v}],[l_{v},``Null"],[l_{v},l_{v}+1]\}]\\ &=\Pr[B \wedge \lnot A ] \\ 
&=\Pr[B]-\Pr[B\wedge A]  \\  &> 1-1/exp(\Omega (N))
\end{align*}
\end{proof}

\begin{lemma}
\label{firstextra}
Let $\ket{\psi(\theta_{v})}$ be a voting state with $\delta \in [0,2\pi/D)$ and $l_{v}=D-1$,where $\delta$ is a  continuous random variable .Then it holds: $$\Pr[x_{D-2}<\Theta_{D,\delta}^{v}<x_{D}]+Pr[x_0<\Theta_{D,\delta}^{v}<x_1]>0.9$$

\end{lemma}

\begin{proof}
\begin{align}
&\Pr[x_0<\Theta_{D,\delta }^{v}<x_1]\\ &= 1/(2\pi D) \int_{x_0}^{x_1}(\dfrac{Sin[D/2(\theta - \theta_v)]}{Sin[1/2(\theta - \theta_v)]})^2 d\theta
\label{apen1}
\end{align}

Now we set $\theta=\theta -x_D$ to \ref{apen1} and we have:

\begin{align}
&\Pr[x_0<\Theta_{D,\delta }^{v}<x_1]\\ &= 1/(2\pi D) \int_{x_D}^{x_D+x_1}(\dfrac{Sin[-D\pi+D/2(\theta - \theta_v)]}{Sin[-\pi+1/2(\theta - \theta_v)]})^2 d\theta \\ &= 1/(2\pi D) \int_{x_D}^{x_D+x_1}(\dfrac{Sin[D/2(\theta - \theta_v)]}{Sin[1/2(\theta - \theta_v)]})^2 d\theta
\end{align}
Finally we have:
\begin{align}
&\Pr[x_{D-2}<\Theta_{D,\delta}^{v}<x_{D}]+\Pr[x_0<\Theta_{D,\delta}^{v}<x_1] \\ &=1/(2\pi D) \int_{x_{D-2}}^{x_D+x_1}(\dfrac{Sin[D/2(\theta - \theta_v)]}{Sin[1/2(\theta - \theta_v)]})^2 d\theta \\ &= 1/(2\pi D)\int_{-2\pi/D}^{4\pi/D} \dfrac{Sin^2[D(\theta-\delta)/2]}{Sin^2[(\theta-\delta)/2]}d\theta
\end{align}

From lemma \ref{proof0.9} this integral is at least $0.9$

\par 
The proof is similar for $l_v=0$.

\end{proof}
\begin{lemma}
Let $\mathtt{Solution}$ be the matrix of algorithm \ref{adval}, then it holds: 
\begin{align*}
&\Pr[\mathtt{Solution} \in \{\{l_{v}-1,l_{v}\},\{l_{v}\},\{l_{v},l_{v}+1\}\}]\\ &=\Pr[\mathtt{Solution} \in \{[l_{v}-1,l_{v}],[l_{v}],[l_{v},l_{v}+1]\}]
\end{align*}
\end{lemma}

\begin{proof}
(sketch)It holds that:
\begin{align}
&\Pr[\mathtt{Solution} \in \{\{l_{v}-1,l_{v}\},\{l_{v}\},\{l_{v},l_{v}+1\}\}]\\ &=\Pr[\mathtt{Solution} \in \{l_{v}-1,l_{v}\}]\\ &+\Pr[\mathtt{Solution} \in \{l_{v}\}]\\ &+\Pr[\mathtt{Solution} \in \{l_{v},l_{v}+1\}]
\end{align}

We need to prove that:
\begin{align}
\Pr[\mathtt{Solution} \in \{l_{v}-1,l_{v}\}]=\Pr[\mathtt{Solution}=[l_{v}-1,l_{v}]]
\end{align}
From the construction of the algorithm \ref{adval} we know that:
\begin{align}
\Pr[\mathtt{Solution}=[l_{v},l_{v}-1]|\mathtt{Solution} \in \{l_{v}-1,l_{v}\}]=0
\end{align}
This is true because the values of the $\mathtt{Solution}$ are from the matrix $\mathtt{Record}$ in a progressive manner. So under the assumption that both $l_{v},l_{v}-1$ had appeared at least 40\% times, they inserted in a progressive order. The only time they will not is the case in which $l_{v}=0$ and $l_{v}-1=D-1$. At first the order is $[0,D-1]$, but because of the special condition we had in our algorithm the order switches to $[D-1,0]$.
\par It holds that:
\begin{align}
&\Pr[\mathtt{Solution}=[l_{v},l_{v}-1]|\mathtt{Solution} \in \{l_{v}-1,l_{v}\}]\\ &=\Pr[\mathtt{Solution}=[l_{v},l_{v}-1]]+\Pr[\emptyset]\\ &=\Pr[\mathtt{Solution}=[l_{v},l_{v}-1]]\\ &=0
\end{align}
Similar are the other cases.
\end{proof}

\end{document}